\documentclass[10pt,twocolumn]{IEEEtran}
\usepackage{amsmath, graphics,amssymb,epsfig,subfigure,color,amsthm,cite}

\newtheorem{theorem}{Theorem}

\newtheorem{lemma}{Lemma}

\newtheorem{corollary}{Corollary}

\begin{document}
 
\title{{ Space-Time Physical-Layer Network Coding} 
}

\author{\normalsize
  \IEEEauthorblockN{Namyoon Lee~\IEEEmembership{Student Member,~IEEE,} and Robert W. Heath Jr.~\IEEEmembership{Fellow,~IEEE,} \thanks{N. Lee and R. W. Heath Jr. are with the Wireless Networking and Communications Group, Department of Electrical and Computer Engineering, The University of Texas at
Austin, Austin, TX 78712, USA. (e-mail:\{namyoon.lee, rheath\}@utexas.edu)}   } 
}
 
%
%
%\\ \small
%  \IEEEauthorblockA{ Wireless Networking and Communications Group\\
% Department of Electrical and Computer Engineering
%\\  The University of Texas at\thanks{The pre-version of this paper was uploaded in Arxiv: http://arxiv.org/pdf/1302.0749v1.pdf. }
%Austin, Austin, TX 78712 USA\\
%    E-mail~:~namyoon.lee@utexas.edu and rheath@utexas.edu}

%% Create the title:
\maketitle

\begin{abstract}
A space-time physical-layer network coding (ST-PNC) method is presented for information exchange among multiple users over fully-connected multi-way relay networks. 
The method involves two steps: i) side-information learning and ii) space-time relay transmission. In the first step, different sets of users are scheduled to send signals over networks and the remaining users and relays overhear the transmitted signals, thereby learning the interference patterns. In the second step, multiple relays cooperatively send out linear combinations of  signals received in the previous phase using space-time precoding so that all users efficiently exploit their side-information in the form of: 1) what they sent and 2) what they overheard in decoding. This coding concept is illustrated through two simple network examples. It is shown that ST-PNC improves the sum of degrees of freedom (sum-DoF) of the network compared to existing interference management methods. With ST-PNC, the sum-DoF of a general multi-way relay network without channel knowledge at the users is characterized  in terms of relevant system parameters, chiefly the number of users, the number of relays, and the number of antennas at relays. A major implication of the derived results is that efficiently harnessing both transmitted and overheard signals as side-information brings significant performance improvements to fully-connected multi-way relay networks.

\end{abstract}

\section{Introduction}

%\subsection{Background}
Interference is a fundamental bottleneck in wireless communication networks whose spectrum is shared among multiple users. Unmanaged interference results in diminishing data rates in wireless networks. With a recently developed network coding strategy, however, it was demonstrated that interference is no longer adverse in communication networks, provided that it can sagaciously be harnessed. This approach of exploiting interference has opened the possibility of better performance in the interference-limited communication regime than was previously thought possible. For example, in multi-hop wired networks, it was shown that a network coding strategy achieves the capacity of the multicast network \cite{Li}. Physical layer (analog) network coding \cite{Wu,Larsson,Zhang,Popovski,Katti,Rankov} provides a generalization of network coding to wireless networks. 
In certain topologies, it was shown that physical layer network coding can achieve higher rates over routing-based strategies.

In this paper, we advance the idea of interference exploitation. The prior physical-layer network coding approaches exploit the self-interference signal as the main source of side-information. We introduce a new physical-layer network coding strategy, which exploits overheard interference signals as side-information in addition to self-interference signals for fully-connected multi-way relay networks.

\textbf{Related Work:}
Multi-way communication using intermediate relay nodes is found in cellular networks, sensor networks, and device-to-device (D2D) communication. The simplest multi-way relay network model is the two-way relay channel \cite{Wu,Larsson,Zhang,Popovski,Katti,Rankov} where a pair
of users wish to exchange messages by sharing a single relay. Although the capacity of this simple channel is still unknown in general \cite{Nam}, physical layer network coding \cite{Wu,Larsson,Zhang,Popovski} and analog network coding \cite{Katti,Rankov} are key techniques for showing how to improve the sum-rates of two-way relay channels by allowing users to exploit their transmit signal as side-information. The two-way relay channel has been generalized in a number of ways to consider multiple users \cite{Chen:09,Sezgin} and multiple directional information exchange \cite{Lee_Lim_Chun:10,Gunduz:09,Chaaban,Ong:11,Lee_Chun:14,Lee_Lee_Lee:12}. For example, for the multi-pair two-way relay channel where multiple user pairs exchange messages with their partners by sharing a common relay, the capacity of multi-pair two-way relay network was characterized for a deterministic and Gaussian channel model in \cite{Sezgin}. For the multi-user multi-way relay channel with unicast messages exchange, the multiple-input multiple-output (MIMO) Y channel was introduced in \cite{Lee_Lim_Chun:10} where three users exchange independent unicast messages with each other via an intermediate relay. The key to deriving the degrees of freedom (DoF) of the MIMO Y channel was the idea of \textit{signal space alignment for network coding}. Subsequently, this idea was applied to characterize the the sum-DoF of a $K$-user Y channel \cite{Lee_Lee_Lee:12} and multi-way MIMO relay channel with asymmetric antennas \cite{Chaaban:13}, mixed (unicast and multicast) information flows \cite{Tian:13}, and direct links between users \cite{namyoon:13}.

The main limitation of the aforementioned studies on multi-way relay channels \cite{Wu,Larsson,Zhang,Popovski,Katti,Rankov, Lee_Lim_Chun:10,Gunduz:09,Chaaban,Ong:11,Lee_Chun:14,Lee_Lee_Lee:12} is that they rely layered network connectivity that ignores direct links among users. For example, in the two-way relay channel  \cite{Wu}-\cite{Rankov}, it was assumed that users cannot communicate with each other without using a relay between them because they are very far apart. Due to the broadcast nature of the wireless medium and the mobility of users, however, it is possible that a wireless node is able to listen to the other node's transmission through a direct path; thereby all nodes in the network can be directly connected with each other. This motivates us to consider a fully-connected multi-way relay network in which $K$ users with a single antenna exchange unicast messages with each other via $L$ relays; each of them has $M_{\ell}\geq 1$ antennas for $\ell\in\{1,2,\ldots,L\}$.

%
%In particular, we assume that all nodes have half-duplex  constraint due to hardware limitations, implying that transmission and reception occurs in different orthogonal time slots. 

\textbf{Contribution:}
The completely-connected property of the multi-way relay networks brings a new challenge in managing interference. When networks are fully-connected, a node receives signals arriving from different paths, which creates a more sophisticated interference management problems than those of partially connected networks. To overcome this challenge, we propose a new interference management approach inspired by physical-layer network coding, called space-time physical layer network coding (ST-PNC). ST-PNC involves two key steps: 1) side-information learning and 2) space-time relay transmission. In the first phase of side-information learning, subsets of users in the network spread out information symbols in the network over multiple time slots. Then, the non-transmitting nodes in the network overhear the information symbols sent by the multiple transmitting nodes and store linear combinations of them to exploit later in decoding. In the second phase, relays send out the superposition of obtained symbols using space-time precoding over multiple channel uses. The core concept of space-time precoding at the relays is to effectively control the multi-directional information flows so that all users can exploit their side-information: i) what they sent and ii) what they overheard in the phase of side-information learning.
%
% which involves two key phases: 1) side-information learning phase and 2) space-time physical-layer network coding phase. In the first phase called side-information learning phase, different subsets of nodes send the messages to the network. During this phase, since the network is fully-connected, the non-transmitting nodes and relays overhear the linear combinations of transmitted messages sent by subsets of nodes. These overheard linear combinations are stored as side information together with what each node sent. In the second phase, the relays send out a superposition of received signals from the previous phase using space-time precoding so that all nodes efficiently exploit their side- information: 1) what they sent to and 2) what they observed through networks. 

We explain the concept of ST-PNC using two simple fully-connected multi-way relay networks. In those networks, it was shown that ST-PNC provides increased sum-DoF of the networks compared to a relay-aided multi-user precoding technique \cite{Tannious_Nosratinia:08} and interference alignment \cite{Tian}. From this result, we verify the intuition that efficiently harnessing both transmitted and overheard signals as side-information brings significant performance improvements to fully-connected multi-way relay networks. Then, applying ST-PNC and relay-aided interference alignment \cite{Tian}, we establish an inner bound of the sum-DoF for the $K$-user fully-connected multi-way relay network with $L$ relays, each with one or more antennas. One interesting observation obtained from this sum-DoF characterization is that if there are not enough antennas at the relays in the multi-way relay network, then the one-way communication protocol method using relay-aided interference alignment achieves a better sum-DoF of the network. Whereas, when the number of antennas at the relays are enough to control multi-directional information flows, the multi-way communication protocol using the proposed ST-PNC outperforms than the existing interference management techniques. Leveraging the cut-set outer bound result in \cite{ Lee_Lim_Chun:10}, we provide a sufficient condition of relays’ antenna configurations for obtaining the optimal sum-DoF of the network. Further, by comparing with a generalization of orthogonalize-and-forward method in \cite{Rankov}, we demonstrate the superiority of the proposed ST-PNC in terms of the sum-DoF.  
The rest of the paper is organized as follows. In Section II,
a general system model of the fully-connected multi-way relay network is described. We illustrate the key idea of the proposed ST-PNC through two simple networks in Section III.  In Section IV, we analyze the sum-DoF of the general fully-connected multi-way relay network. The paper concludes with future directions in Section V.

Throughout this paper, transpose, conjugate transpose, inverse of a matrix ${\bf X}$ are represented by ${\bf X}^{T}$ ,
${\bf X}^{*}$, ${\bf X}^{-1}$, respectively. %In addition, $\mathbb{C}$ and $\mathbb{R}$
%indicates a complex and real value.% $\mathcal{CN}(0,1)$
%represents a complex Gaussian random variable with zero mean and unit variance. 
%

\section{System Model}
\begin{figure}
\centering
\includegraphics[width=3.0in]{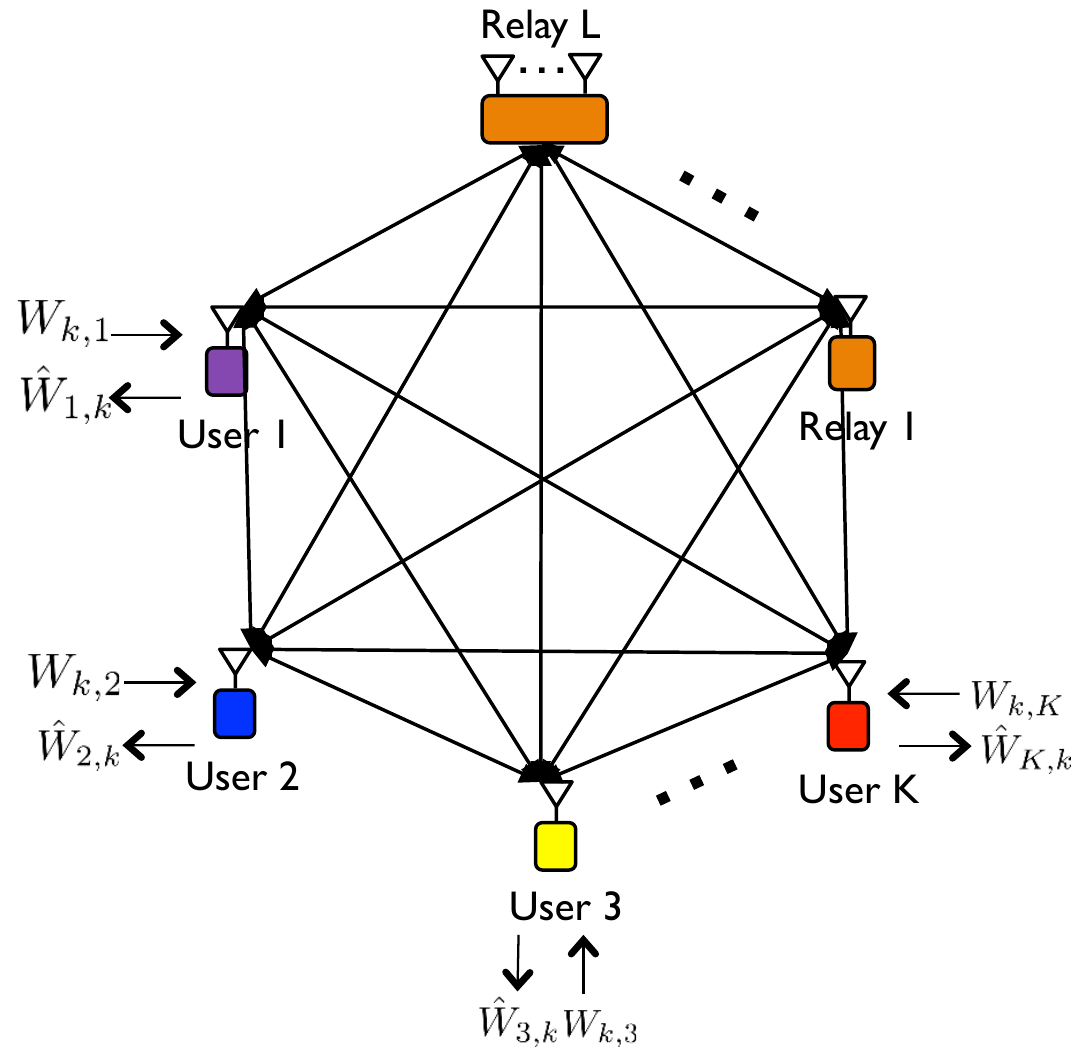}\vspace{-0.1cm}
\caption{A $K$-user fully-connected multi-way relay network with $L$ relays each of which has $M_{\ell}\geq 1$ antennas. In this network, user $i$ wants to send $K-1$ messages $W_{k,i}$ and decode ${\hat W}_{i,k}$ for $k\in \mathcal{U}/\{i\}$ by sharing the multiple relays. This fully-connected multi-way communication network can model various applications for data sharing among D2D users or sensors. } \label{fig:1}\vspace{-0.1cm}
\end{figure}

Let us consider a network comprised of $K$ users each with a single antenna and $L$ relays each of which has $M_{\ell}$ antennas. In this network, each user wants to exchange $K-1$ unicast messages with every other user. Denoting sets $\mathcal{U}= \{1,2,\ldots,K\}$ and $\ \mathcal{U}_k^{\textrm{c}}= \{1,2,\ldots,K\}/\{k\}$, user $k\in\mathcal{U}$ desires to send $K-1$ unicast messages $W_{i,k}\in\{1,2,\ldots,2^{nR_{i,k}}\}$ for $i\in \mathcal{U}/\{k\}$ for user $i$ and intends to decode $K-1$ messages $W_{k,i}$ for $i \in \mathcal{U}_k^{\textrm{c}}$ sent by all other users.
We assume that the network is completely-connected as illustrated in Fig. \ref{fig:1}, implying that any node can communicate with any a other node through a direct path in the network. Further, we assume that all nodes operate in half-duplex mode, i.e., transmission and reception span orthogonal time slots. User $k\in\mathcal{U}$ generates a sequence of transmit signals $\{x_k[t]\}_{t=1}^n= f_k(W_{1,k},\ldots,W_{k\!-\!1,k},W_{k\!+\!1,k},\ldots,W_{K,k})$ using a ``restricted encoder" $f_k(\cdot)$ which does not use the previously received channel output but it only exploits the transmit messages in encoding.

Let $\mathcal{S}_t$ and $\mathcal{D}_t$ denote the set of source and destination nodes in time slot $t$. Due to the fully-connected property and the half-duplex constraint, when the users in $\mathcal{S}_t$ simultaneously send their signals in time slot $t$, the received signals at user $k\in \mathcal{D}_t$, $y_k[t]$, and relay $\ell \in \{1,2,\ldots,L\}$, ${\bf y}^{\ell}_{{\rm R}}[t]\in \mathbb{C}^{M_{\ell} \times 1} $, are given by
\begin{align}
{ y}_{k}[t] =&\sum_{i \in \mathcal{S}_t}{h}_{k,i}[t]s_{k,i}+ {z}_{k}[t], \quad  k\in\mathcal{D}_t,\\
{\bf y}^{\ell}_{{\rm R}}[t]=&\sum_{i \in \mathcal{S}_t}{\bf h}^{\ell}_{{\rm R},i}[t]s_{k,i}+{\bf z}_{{\rm R}}^{\ell}[t],
\end{align}
where ${z}_k[t]$ and ${\bf z}_{\rm R}^{\ell}[t]$ denote the additive noise signal at user $k$ and at relay $\ell$ in time slot $t$ whose elements
are Gaussian random variables with zero mean and unit variance,
i.e., $\mathcal{CN}(0,1)$, and ${{h}_{k,i}}[t]$ and ${\bf h}_{{\rm R},i}^{\ell}[t]=\left[ {h}^{\ell,1}_{{\rm R},i}[t], \ldots,  {h}^{\ell,M_{\ell}}_{{\rm R},i}[t]\right]^T$ represent the channel coefficients from user $i$ to user $k$ and the channel vector from user $i$ to relay $\ell$, respectively.

When the relay and user $i \in S_t$ cooperatively transmit in time slot $t$, at the same time, user $k\in \mathcal{D}_t$ receives the signal as
\begin{align}
{ y}_{k}[t]= \sum_{i \in \mathcal{S}_t}{h}_{k,i}[t]x_{i}[t]+ {{\bf h}^{\ell}_{k,{\rm R}}[t]}^{*}{\bf x}_{{\rm R}}^{\ell}[t]+ {z}_{k}[t],\quad k\in\mathcal{D}_t,
\end{align}
where ${{\bf h}^{\ell}_{j,{\rm R}}[t]}^{*}=\left[ {h}^{\ell,1}_{j,{\rm R}}[t], \ldots,  {h}^{,\ell, M_{\ell}}_{j,{\rm R}}[t] \right] \in \mathbb{C}^{1\times M_{\ell}}$ denotes the (downlink) channel vector from relay $\ell$ to user $k$ and ${\bf x}^{\ell}_{{\rm R}}[t]$ represents the transmit signal vector at relay $\ell$ when the $t$-th channel is used.

The transmit power at each user and the relay is assumed to satisfy the power constraints,
$\frac{1}{n}\sum_{t=1}^n\mathbb{E}\left[|{x}_{i}[t]|^2\right] \leq P$ and $\frac{1}{n}\sum_{t=1}^n\mathbb{E}\left[\|{\bf x}^{\ell}_{{\rm R}}[t]\|_2^2\right] \leq P$. Further, the entries of all
channel elements of ${h}_{k,i}[t]$, ${\bf h}^{\ell}_{{\rm R},i}[t]$, and ${{\bf h}^{\ell}_{k,{\rm R}}[t]}^{*}$ are drawn from an independent and identically distributed (IID) continuous
distribution and their absolute values are bounded between a nonzero minimum value and a finite maximum value. The channel state information (CSI) is assumed to be perfectly known to users and relays in receiving mode for their own channels. Further, relays have global CSI of all channel links in transmitting mode thanks to error-free feedback links, i.e., global channel state information at transmitter (CSIT), while users have no CSIT.

User $k$ sends an independent message  $W_{i,k}$ for one intended user $i$ with rate $R_{i,k}(P)=\frac{\log_2|W_{i,k}|}{n}$ for $i, k \in\mathcal{U}$ and $i \neq k$. Then, rate $R_{i,k}(P)$ is achievable if user $i$ can
decode the desired message with an error probability that is arbitrarily small for sufficient channel uses $n$. The sum-DoF characterizing the approximate sum-rate in the high \textrm{SNR} regime is defined as a function of the number of users and the number of antennas at the relays: 
\begin{align}
d_{\Sigma}(K,\{M_{\ell}\})&=\lim_{{P}\rightarrow \infty}\frac{\sum_{k=1,k\neq i}^{K}\sum_{i=1}^{K}R_{k,i}\left( P\right)}{\log\left({P}\right)}.%\\
%&=\sum_{k=1,k\neq i}^{K}\sum_{i=1}^{K}d_{k,i}.
\end{align}
Using the sum-DoF metric in this paper suitably captures the signal interactions by deemphasizing the effects of noise, thereby providing a clear understanding of the scaling behavior of the sum-capacity for sophisticated networks.

\section{Space-Time Phaysical Layer Network Coding}
In this section, we illustrate the core ideas behind our approach starting with two simple examples. Gaining insights from this section, we extend our method into the general multi-way relay network in the next section.

%\subsection{Concept and Example}
%The proposed space-time physical layer network coding involves two ingredients: 1) side-information learning and 2) space-time relay transmissions. In the phase of side-information learning, subsets of users in the network spread out information symbols to a network over multiple time slots. Then, the non-transmitting nodes in the network overhear the transmitted information symbols and save linear combinations of them to exploit in decoding. After this learning phase, relays send out the superposition of obtained symbols using space-time precoding over multiple channel uses. The core concept of space-time precoding at relays is to effectively control multi-directional information flows so that all users can exploit their side-information, what it sent and what it overheard during the phase of side-information learning. 

\subsection{Example 1: Restricted Two-Pair Two-Way Interference Channel with a MIMO Relay }

\begin{figure}
\centering
\includegraphics[width=3.3in]{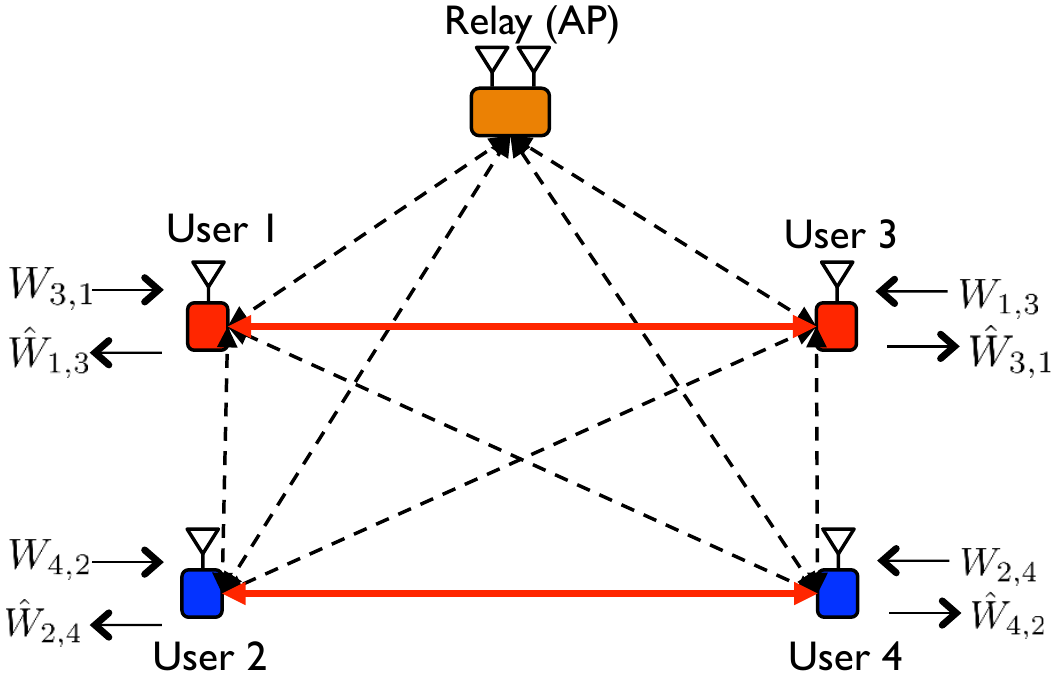}\vspace{-0.1cm}
\caption{The two-pair two-way interference channel with a two-antenna relay. Each user wants to exchange the messages with its partner by using a shared relay.} \label{fig:2}\vspace{-0.1cm}
\end{figure} 

Consider a four-user fully-connected multi-way relay channel with a multi-antenna relay. As illustrated in Fig. \ref{fig:2}, we set $W_{2,1} =W_{4,1} =\phi$, $W_{1,2} =W_{3,2} =\phi$, $W_{2,3} =W_{4,3} =\phi$, and $W_{1,4} =W_{3,4} =\phi$. In this case, two pairs (user 1-3 and user 2-4) exchange messages with their partners via a relay ($L=1$) with $M_1=2$ antennas. This scenario can model the case where two D2D user pairs cooperatively exchange video files with the help of a multi-antenna base station or access point (AP). When the relay node is discarded, this channel model is equivalent to a two-way interference channel \cite{Suh} but using a restricted encoder. Therefore, we refer to this channel as ``the restricted two-pair two-way interference channel with a MIMO relay." Throughout this example, we will demonstrate the following theorem using the proposed ST-PNC strategy.

 \begin{theorem} \label{Theorem1}
For the restricted two-pair two-way interference channel with a relay employing two antennas, a total $d^{\rm TWIC}_{\Sigma}=\frac{4}{3}$ of sum-DoF is achievable without CSIT at users but with CSIT at the relay.
\end{theorem}

\begin{proof}
ST-PNC involves two phases: side-information learning and space-time relay transmissions. In this proof we explain how each of these phases is exploited by the receiver to achieve the stated sum-DoF. 

\subsubsection{Side-Information Learning} We use two time slots for side-information learning. In the first time slot, user 1 and user 2 send signals $x_1[1]=s_{3,1}$ and $x_2[1]=s_{4,2}$, while user 3, user 4, and the relay listen to the transmitted signals, i.e. $\mathcal{S}_1=\{1,2\}$ and $\mathcal{D}_{1}=\{3,4,{\rm R}_1\}$. Ignoring noise at the receivers, the received signals at user 3, user 4, and the relays are given by
\begin{align}
{ y}_{3}[1] =&{h}_{3,1}[1]s_{3,1}+  {h}_{3,2}[1]s_{4,2}, \label{eq:dis2} \\
{ y}_{4}[1] =& {h}_{4,1}[1]s_{3,1}+  {h}_{4,2}[1]s_{4,2}, \label{eq:dis1}\\
{ \bf y}^1_{\rm{R}}[1]=&{\bf h}^{1}_{\rm{R},1}[1]s_{3,1} +{\bf h}^{1}_{\rm{R},2}[1] s_{4,2}.   \label{eq:dis3}
\end{align} 
In the second time slot, user $3$ and user $4$ send signals ${x}_{3}[2]={ s}_{1,3}$ and ${ x}_{4}[2]={s}_{{2,4}}$ over the backward interference channel. The received signals at user 1, user 2, and the relay are:
\begin{align}
y_1[2]=&{ h}_{1,3}[2]{ s}_{1,3}+{ h}_{1,4}[2]{ s}_{{2,4}} \\
y_2[2]=&{ h}_{2,3}[2]{ s}_{{1,3}}+{ h}_{2,4}[2]{ s}_{{2,4}} \\ 
{\bf y}^1_{\rm{R}}[2]=&{ {\bf h}}^{1}_{\rm{R},3}[2]{ s}_{1,3} +{ { \bf h}}^{1}_{\rm{R},4}[2] { s}_{{2,4}}.
\end{align}
During side-information learning, each user obtains one linear equation that contains the desired information symbol. Further, the relay acquires four equations that contain all the information symbols in the network. Under the noiseless assumption, by using a zero-forcing (ZF) decoder, the relay perfectly decodes four information symbols from the four equations, providing the knowledge of $s_{3,1}$, $s_{1,3}$, $s_{4,2}$, and $s_{2,4}$.

\subsubsection{Space-Time Relay Transmission} We use the third time slot for the space-time relay transmission. The relay sends a linear combination of the received signals during the previous phase using the space-time precoding matrix ${\bf V}^{\ell}_{\rm R}[3]=\left[{\bf v}_{3,1}[3],~{\bf v}_{1,3}[3],~{\bf v}_{4,2}[3],~{\bf v}_{2,4}[3]\right]\in\mathbb{C}^{2\times 4}$. The transmitted signal of the relay in time slot 3 is
 \begin{align}
{ \bf x}_{\rm R}[3]&={\bf v}_{3,1}[3]s_{3,1} +{\bf v}_{1,3}[3]s_{1,3}+{\bf v}_{4,2}[3]s_{4,2}+{\bf v}_{2,4}[3] s_{2,4}.%s_{4,2}+z_{\rm R}^{\ell}[1]\right) +v^{\ell}_{\rm R}[3,2]\left({ { h}}^{\ell}_{\rm{R},3}[2]{ s}_{1,3} +{ { h}}^n_{\rm{R},4}[2] { s}_{{2,4}}+z_{\rm R}^{\ell}[2]\right),
 \end{align}
Then, the received signal at user $j\in\{1,2,3,4\}$ in time slot 3 from the relay transmissions is
\begin{align}
{y}_{j}[3]&= {{\bf h}^{1}_{j,\rm{R}}}^{\!\!*}[3]{\bf x}_{\rm{R}}[3]  \nonumber\\
&=  {{\bf h}^{1}_{j,\rm{R}}}^{\!\!*}[3]\left({\bf v}_{3,1}[3]s_{3,1} +{\bf v}_{1,3}[3]s_{1,3}\right) \nonumber\\
&+ {{\bf h}^{1}_{j,\rm{R}}}^{\!\!*}[3]\left({\bf v}_{4,2}[3]s_{4,2}+{\bf v}_{2,4}[3] s_{2,4}\right).\label{eq:rec_two_IC}
\end{align}
The key idea of the space-time relay transmission is to control interference propagation on the network so that each user can exploit two types of side-information: i) what it transmitted and ii) what it overheard during the previous phase. For example, user 1 wants to decode data symbol $s_{1,3}$ and has two different forms for side-information: transmitted symbol $s_{3,1}$ in time slot 1 and the received signal in time slot 2, i.e., $y_1[2] = h_{1,3}[2]s_{1,3} + h_{1,4}[2]s_{2,4}$. To exploit both different types of side-information simultaneously, the relay should not propagate interference symbol $s_{4,2}$ to user 1 by selecting $ {\bf v}_{4,2}[3] \in \rm{null}({{\bf h}^{1}_{1,\rm{R}}}^{\!\!*}[3])$, as it is \textit{unmanageable interference} to user 1. Similarly, to make all users harness their side-information, the relay needs to cancel unmanageable interference signals by constructing the space-time relay precoding vectors such that 
\begin{align}
 {{\bf h}^{1}_{2,\rm{R}}}^{\!\!*}[3]{\bf v}_{3,1}[3]=0,~ {{\bf h}^{1}_{3,\rm{R}}}^{\!\!*}[3]{\bf v}_{2,4}[3]=0,~{{\bf h}^{1}_{4,\rm{R}}}^{\!\!*}[3]{\bf v}_{1,3}[3]=0.
\label{eq:IN_two_IC}
\end{align}
Since the precoding solutions for ${\bf v}_{i,j}[3]$ always exist in this case because of the existence of null space of the ${{\bf h}^{1}_{j,\rm{R}}}^{\!\!*}[3]$ in (\ref{eq:IN_two_IC}), it is possible to control undesired interference signals from the relay transmission.%Thus, the received signals at users in time slot 3 are given by
%\begin{align}
%{y}_{1}[3]
%&={\bf h}^*_{1,\textrm{R}}[3]{\bf v}_{1,3}[3]s_{1,3}+{\bf h}^*_{1,\textrm{R}}[3]{\bf {v}}_{2,4}[3]{s}_{2,4}+\underbrace{{\bf h}^*_{1,\textrm{R}}[3]{\bf { v}}_{3,1}[3]{ s}_{3,1}}_{\textrm{Self-interference}}, \\ 
%{y}_{2}[3]&={\bf h}^*_{2,\textrm{R}}[3]{\bf v}_{1,3}[3]s_{1,3}+{\bf h}^*_{2,\textrm{R}}[3]{\bf {v}}_{2,4}[3]{ s}_{2,4}+\underbrace{{\bf h}^*_{2,\textrm{R}}[3]{\bf { v}}_{4,2}[3]{ s}_{4,2}}_{\textrm{Self-interference}}, \\
%{{y}}_{3}[3]&={\bf {h}}^*_{3,\textrm{R}}[3]{\bf v}_{3,1}[3]s_{3,1}+{\bf h}^*_{3,\textrm{R}}[3]{\bf v}_{4,2}[3]s_{4,2}+\underbrace{{\bf {h}}^*_{3,\textrm{R}}[3]{\bf v}_{1,3}[3]{s}_{1,3}}_{\textrm{Self-interference}}, \\  
%{y}_{4}[3]&={\bf h}^*_{4,\textrm{R}}[3]{\bf v}_{3,1}[3]s_{3,1}+{\bf h}^*_{4,\textrm{R}}[3]{\bf v}_{4,2}[3]s_{4,2}+\underbrace{{\bf h}^*_{4,\textrm{R}}[3]{\bf {v}}_{2,4}[3]{s}_{2,4}}_{\textrm{Self-interference}}.
%\end{align}

\subsubsection{Decoding} Successive interference cancellation is used to eliminate the back propagating self-interference from the received signal in time slot 3. The remaining inter-user interference is removed by a ZF decoder. For instance, the received signal of user $1$ in time slot 3 is 
\begin{align}
{y}_{1}[3]&={{\bf h}^{1}_{1,\textrm{R}}}^{\!\!*}[3]{\bf v}_{1,3}[3]s_{1,3}+{{\bf h}^{1}_{1,\textrm{R}}}^{\!\!*}[3]{\bf {v}}_{2,4}[3]{s}_{2,4}\\\nonumber 
&+\underbrace{{{\bf h}^{1}_{1,\textrm{R}}}^{\!\!*}[3]{\bf { v}}_{3,1}[3]{ s}_{3,1}}_{\textrm{Self-interference}}.
\end{align}
Eliminating self-interference ${{\bf h}^{1}_{1,\textrm{R}}}^{\!\!*}[3]{\bf { v}}_{3,1}[3]{ s}_{3,1}$ from $y_1[3]$ as it is known to user 1, we have
 \begin{align}
&{y}_{1}[3]-{{\bf h}^{1}_{1,\textrm{R}}}^{\!\!*}[3]{\bf { v}}_{3,1}[3]{ s}_{3,1}\\\nonumber
&={{\bf h}^{1}_{1,\textrm{R}}}^{\!\!*}[3]{\bf v}_{1,3}[3]s_{1,3}+{{\bf h}^{1}_{1,\textrm{R}}}^{\!\!*}[3]{\bf {v}}_{2,4}[3]{s}_{2,4}.
\end{align}
Concatenating the received signals in time slot 2 and 3, the effective input-output relationship is 
\begin{align}
&\left[%
\begin{array}{c}
  {y}_{1}[2]\\
 {y}_{1}[3]-{{\bf h}^{1}_{1,\textrm{R}}}^{\!\!*}[3]{\bf { v}}_{3,1}[3]{ s}_{3,1}\\
\end{array}%
\right]\\\nonumber&=\underbrace{\left[%
\begin{array}{cc}
 {h}_{1,3}[2] & { h}_{1,4}[2] \\
{{\bf h}^{1}_{1,\textrm{R}}}^{\!\!*}[3]{\bf {v}}_{1,3}[3] & {{\bf h}^{1}_{1,\textrm{R}}}^{\!\!*}[3]{\bf v}_{2,4}[3]\\
\end{array}%
\right]}_{{\bf \tilde{H}}^{\rm TWIC}_{1}}\left[%
\begin{array}{c}
  { s}_{1,3} \\
 { s}_{2,4} \\
\end{array}%
\right]. \label{eq:rec_two_Ic}
\end{align}
Since precoding vectors ${\bf {v}}_{1,3}[3]$ and ${\bf { v}}_{2,4}[3] $ were designed independently of $ {{\bf h}_{1,\textrm{R}}}^{\!\!*}[3]$ and all channel coefficients were drawn from a continuous random distribution, the effective channel matrix ${\bf \tilde{H}}^{\rm TWIC}_{1}$ has a rank of two almost surely. This implies that it is possible to decode desired symbol ${ s}_{1,3}$ by applying a ZF decoder that eliminates the effect of inter-user interference ${ s}_{2,4}$. By symmetry, the other users are able to decode their desired symbols with the same decoding procedure. As a result, a total $4$ of the independent data symbols are delivered over three orthogonal channel uses, achieving a sum-DoF of $d^{\rm TWIC}_{\Sigma}=\frac{4}{3}$. 
\end{proof}

\textit{Remark 1 (Sum-DoF Gains)}: 
To see how the proposed method is useful in terms of sum-DoF, it is instructive to compare our result with other transmission methods. In the two-pair two-way interference channel with a two-antenna relay, there are two interesting candidates.
\begin{itemize}
\item  Time-division-multiple-access (TDMA): As a baseline method, TDMA can be applied, in which one user sends a signal to the corresponding user through a direct link per time slot. This method achieves a sum-DoF of one.%, which is the information-theoretically optimal when source and destination is separate.

\item Multi-user MIMO transmission in \cite{Tannious_Nosratinia:08}: Instead of using direct paths between users, one may also consider two-hop multi-user MIMO transmission in which two users simultaneously send information symbols to the relay in the first hop and the relay broadcasts two symbols using multi-user precoding, eliminating inter-user interference in the second hop. Since four time slots are required to exchange four information symbols, this method also achieves a sum-DoF of one. 
\end{itemize}
This comparison reveals that our strategy exploiting overheard signals as side-information provides at least a $33\%$ better sum-DoF than other reasonable methods in this particular network scenario.

\textit{Remark 2 (CSI Knowledge and Feedback)}: To cancel self-interference, it is assumed that each user has knowledge of the effective channel from the relay to users (e.g. ${{\bf h}^1_{1,\textrm{R}}}^{\!\!*}[3]{\bf { v}}_{3,1}[3] $ for user 1). This effective channel, however, can be estimated using demodulation reference signals in commercial wideband systems; thereby users do not need to know CSIT, i.e., no CSI feedback is required between users. In contrast, the relay needs to know CSIT from the relay to users, i.e., local CSIT, to construct the precoding vectors. While this CSIT is possibly obtained by a feedback link if the frequency division duplexing system is considered, it also can be acquired without feedback using time division duplex system thanks to channel reciprocity.

%
%
%\textit{Remark 3 (Link Diversity)}: As shown in (\ref{eq:rec_two_Ic}), the proposed ST-PNC allows to receive the desired signal from two different paths: a direct path and a detoured path via the relay. Thus, the user can benefit from link (channel) diversity when it decodes the desired signal. 

\subsection{Example 2: Restricted Two-Pair Two-Way Restricted X Channel with a MIMO Relay}
Let us consider the same physical channel model as in Example 1 but a more complex information exchange scenario. In this example, as depicted in Fig. \ref{fig:two_x}., user 1 and user 2 intend to exchange two independent messages with both user 3 and 4. Since this channel model can be viewed as a bi-directional X channel as in \cite{two_way_X}, we refer to this scenario as a restricted two-pair two-way X channel with a multiple antenna relay. Note that this setup is a special case of the 4-user multi-way relay network in which $W_{2,1}=W_{1,2}=\phi$ and $W_{3,4}=W_{4,3}=\phi$. In this example, we will prove the following theorem.
 \begin{theorem} \label{Theorem2}
For the restricted two-pair two-way X channel with a relay employing two antennas, a total $d^{\rm TWXC}_{\Sigma}=\frac{8}{5}$ of sum-DoF is achievable without CSIT at the users but with global CSIT at the relay.
\end{theorem}

\begin{figure}
\centering
\includegraphics[width=3.5in]{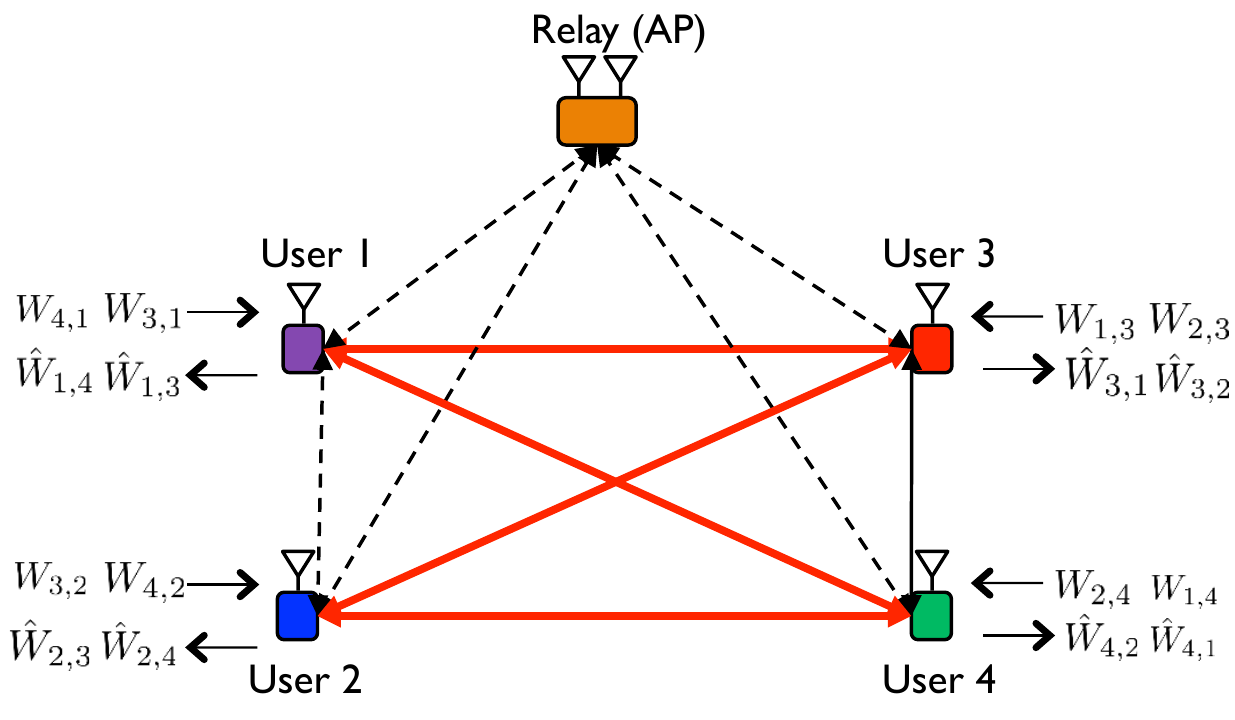}\vspace{-0.2cm}
\caption{The restricted two-pair two-way X channel with a two antennas relay. Each user wants to exchange two independent messages with the other user group by using a shared relay.} \label{fig:two_x} \vspace{-0.1cm}
\end{figure}

%In this example, we consider $K=4$ and $N=2$ case where user 1 and 2 want to exchange two independent messages with both user 3 and 4. As a special case of the fully-connected multi-way relay channel, 

\begin{proof}
We prove Theorem 2 with the proposed ST-PNC strategy. %Since in the example we consider a more sophisticated information exchange scenario than Example 1, we combine the idea of interference alignment with interference neutralization in the phase of space-time relay transmissions to control information pro , which requires global channel knowledge at the relay.

\subsubsection{Side-Information Learning Phase} This phase consists of four time slots.
During the first two time slots, user 1 and user 2 become transmitting nodes while the other nodes listen to the transmitted symbols, i.e., $\mathcal{S}_t=\{1,2\}$ and $\mathcal{D}_t=\{\rm{R}_1,3,4\}$ for $t\in\{1,2\}$. In the first time slot, user 1 and user 2 send an independent
symbol intended for user $3$, i.e., ${x}_{1}[1]\!=\!s_{3,1}$ and
${x}_{2}[1]\!=\!s_{3,2}$. In the second time slot, they send independent
symbols intended for user $4$, i.e., ${x}_{1}[2]\!=\!s_{4,1}$ and
${x}_{2}[2]\!=\!s_{4,2}$. Neglecting noise at the receivers, user $3$ and user $4$ obtain two linear equations during two time slots, which are
\begin{align}
{y}_{3}[1] &={h}_{3,1}[1]s_{3,1}+  {h}_{3,2}[1]s_{3,2} \\
{ y}_{4}[1] &={h}_{4,1}[1]s_{3,1}+  {h}_{4,2}[1]s_{3,2},  \\
{ y}_{3}[2] &= {h}_{3,1}[2]s_{4,1}+  {h}_{3,2}[2]s_{4,2}, \\
{ y}_{4}[2]&= {h}_{4,1}[2]s_{4,1}+  {h}_{4,2}[2]s_{4,2}.
\end{align}

For $t\in\{3,4\}$, the role of transmitters and receivers is reversed, i.e., $\mathcal{S}_t=\{3,4\}$ and $\mathcal{D}_t=\{1,2\}$. In time slot 3, user $3$ and user $4$ send an independent symbol intended for user 1, ${ {x}}_{1}[3]\!=\!{ s}_{1,3}$ and ${{x}}_{2}[3]\!=\!{ s}_{1,4}$. For time slot 4, user $3$ and user $4$ deliver information symbols intended for user 2, ${x}_{3}[4]\!=\! {s}_{2,3}$ and ${x}_{4}[4]\!=\!{ s}_{2,4}$. Therefore, user 1 and user 2 obtain two equations during phase two, which are given by
\begin{align}
{ y}_{1}[3] &={h}_{1,3}[3]s_{1,3}+  {h}_{1,4}[3]s_{1,4} \\  
{ y}_{2}[3] &= {h}_{2,3}[3]s_{1,3}+  {h}_{2,4}[3]s_{1,4},  \\
{ y}_{1}[4] &= {h}_{1,3}[4]s_{2,3}+  {h}_{1,4}[4]s_{2,4}, \\
{ y}_{2}[4] &= {h}_{2,4}[4]s_{2,3}+  {h}_{2,4}[4]s_{2,4}.  %\nonumber\\
%{\bf y}_{{\rm R}}[1]&=&{\bf h}_{{\rm R},1}[1]s_{1} +{\bf h}_{{\rm R},2}[1] s_2, \label{eq:example1}
\end{align} 
Due to the broadcast nature of the wireless medium, the relay is also able to listen to the transmissions by the users. Since it has two antennas, the relay resolves two symbols in each transmission, yielding the full knowledge of $\{ { s}_{1,3}, { s}_{1,4}, { s}_{2,3}, {s}_{2,4},s_{3,1},s_{3,2},s_{4,1}.s_{4,2}\}$. 

%
%Since the relay has two antennas, under the noiseless assumption, it decodes four data symbols ${ s}_{1,3}$, ${ s}_{1,4}$, ${ s}_{2,3}$, and ${s}_{2,4}$ over two times by utilizing a ZF decoder.

%Since the relay exploits knowledge of the current downlink CSI from the relay to the users, i.e., ${\bf h}_{k,{\rm R}}[5]$ for $k \in \{1,2,3,4\}$ and outdated CSI between users i.e., $\left\{h_{4,1}[1],h_{4,2}[1],h_{3,1}[2],h_{3,1}[2]\right\}$ as well as the outdated CSI between the users, i.e., $\left\{{ h}_{2,3}[3],{h}_{2,4}[3],{ h}_{1,3}[3],{h}_{1,4}[4] \right\}$. Using this information,o as  such that each user can exploit side-information what it transmitted and overheard during the previous phase

\subsubsection{Space-Time Relay Transmission Phase} This phase uses only one time slot. 
In time slot $t=5$, only the relay transmits a signal while the other users listen: $\mathcal{S}_5=\{\rm R_1\}$ and $\mathcal{D}_5=\{1,2,3,4\}$. Specifically, the relay sends the superposition of eight data symbols $\{s_{i,j}, s_{j,i}\}$ for $i\in\{1,2\}$ and $j\in\{3,4\}$, which are acquired during the previous phase, using space-time precoding vectors $\left\{ {\bf v}_{j,i}[5],{\bf v}_{i,j}[5]\right\}$, 
\begin{align}
{\bf x}_{{\rm R}}[5]= \sum_{j=3}^{4}\sum_{i=1}^{2}{\bf v}_{j,i}[5]{s}_{j,i} +  \sum_{j=1}^{2}\sum_{i=3}^{4}{\bf { v}}_{j,i}[5]{{ s}}_{j,i}.
\end{align}
We explain the design principle of ${\bf v}_{3,1}[5]$ carrying $s_{3,1}$ using the idea of the space-time relay transmission. Notice that the data symbol $s_{3,1}$ is only desired by user $3$ and it is interference to all the other users except for user 1. This is because user 1 has already $s_{3,1}$, implying that it can be exploited for self-interference cancellation in decoding. User 4 observed $s_{3,1}$ in time slot 1 in the form of ${y}_{4}[1]=h_{4,1}[1]s_{3,1}+h_{4,2}[1]s_{3,2}$. Therefore, user $4$ can cancel $s_{3,1}$ from the relay transmission if it receives the same interference shape $h_{4,1}[1]s_{3,1}$.  Unlike user 4, $s_{3,1}$ is unmanageable interference to user 2. Thus, the relay must design the beamforming vector carrying $s_{3,1}$ so that it does not reach to user 2 while providing the same interference shape to user 4, 
\begin{align}
{{\bf h}_{2,{\rm R}}^{1}}^{\!\!*}[5]{\bf v}_{3,1}[5]&=0 ~~{\rm and} ~~
{{\bf h}_{4,{\rm R}}^{1}}^{\!\!*}[5]{\bf v}_{3,1}[5]= h_{4,1}[1].
\end{align}
Applying the same design principle, we pick the other precoding vectors so that the following conditions are satisfied as 
\begin{align}
\left[%
\begin{array}{c}
{{\bf h}_{2,{\rm R}}^{1}}^{\!\!*}[5]\\
{{\bf h}_{3,{\rm R}}^{1}}^{\!\!*}[5] \\
\end{array}%
\right]{\bf v}_{4,1}[5]&=
\left[
     \begin{array}{c}
           0\\
	h_{3,1}[2]\\
         \end{array}
       \right], \\
\left[%
\begin{array}{c}
{{\bf h}_{1,{\rm R}}^{1}}^{\!\!*}[5]\\
{{\bf h}_{4,{\rm R}}^{1}}^{\!\!*}[5]\\
\end{array}%
\right]{\bf v}_{3,2}[5]&=
\left[
         \begin{array}{c}
           0\\
	h_{4,2}[1]\\
         \end{array}
       \right],
\end{align}
\begin{align}
\left[%
\begin{array}{c}
{{\bf h}_{1,{\rm R}}^{1}}^{\!\!*}[5]\\
{{\bf h}_{3,{\rm R}}^{1}}^{\!\!*}[5]\\
\end{array}%
\right]{\bf v}_{4,2}[5]&=
\left[
         \begin{array}{c}
           0\\
	h_{3,2}[2]\\
         \end{array}
       \right], \\
\left[%
\begin{array}{c}
{{\bf h}_{4,{\rm R}}^{1}}^{\!\!*}[5] \\
{{\bf h}_{2,{\rm R}}^{1}}^{\!\!*}[5]\\
\end{array}%
\right]{\bf {v}}_{1,3}[5]&=
\left[
     \begin{array}{c}
           0\\
	{ h}_{2,3}[3]\\
         \end{array}
       \right],
\end{align}
\begin{align}
\left[%
\begin{array}{c}
{{\bf h}_{4,{\rm R}}^{1}}^{\!\!*}[5] \\
{{\bf h}_{1,{\rm R}}^{1}}^{\!\!*}[5]\\
\end{array}%
\right]{\bf {v}}_{2,3}[5]&=
\left[
     \begin{array}{c}
           0\\
	{ h}_{1,3}[4]\\
         \end{array}
       \right], \\
\left[%
\begin{array}{c}
{{\bf h}_{3,{\rm R}}^{1}}^{\!\!*}[5] \\
{{\bf h}_{2,{\rm R}}^{1}}^{\!\!*}[5]\\
\end{array}%
\right]{\bf {v}}_{1,4}[5]&=
\left[
         \begin{array}{c}
           0\\
	{ h}_{2,4}[3]\\
         \end{array}
       \right],
\end{align}
\begin{align}
\left[%
\begin{array}{c}
{{\bf h}_{3,{\rm R}}^{1}}^{\!\!*}[5] \\
{{\bf h}_{1,{\rm R}}^{1}}^{\!\!*}[5]\\
\end{array}%
\right]{\bf { v}}_{2,4}[5]=
\left[
         \begin{array}{c}
           0\\
	{ h}_{1,4}[4]\\
         \end{array}
       \right].
\end{align}
To implement this, the relay should have current CSIT, i.e., ${{\bf h}^1_{k,{\rm R}}}^{\!\!*}[5]$ for $k \in \{1,2,3,4\}$ and outdated CSI between users i.e., $\left\{h_{4,1}[1],h_{4,2}[1],h_{3,1}[2],h_{3,2}[2],{ h}_{2,3}[3],{h}_{2,4}[3],{ h}_{1,3}[3],{h}_{1,4}[4] \right\}$. Since we assume that the channel coefficients are drawn from a continuous distribution, it is always possible to obtain the solution of ${\bf v}_{i,j}[5]$. From this relay transmission, each user acquires an equation that consists of three sub-equations, each of which corresponds to desired, self-interference, and aligned-interference equations. For instance, the received signal at user 1 is given by
\begin{align}
{y}_{1}[5]=&{{\bf h}_{1,{\rm R}}^{1}}^{\!\!*}[5]\left(\sum_{j=3}^{4}\sum_{i=1}^{2}{\bf v}_{j,i}[5]{s}_{j,i} +  \sum_{j=1}^{2}\sum_{i=3}^{4}{\bf { v}}_{j,i}[5]{{ s}}_{j,i}\right),\nonumber \\
=&\underbrace{ ({{\bf h}_{1,{\rm R}}^{1}}^{\!\!*}[5]{\bf v}_{1,3}[5]){s}_{1,3}+({{\bf h}_{1,{\rm R}}^{1}}^{\!\!*}[5]{\bf v}_{1,4}[5]){s}_{1,4}}_{L_{1, {\rm D}}[5]}\nonumber \\
+&\underbrace{({{\bf h}_{1,{\rm R}}^{1}}^{\!\!*}[5]{\bf {v}}_{3,1}[5]){{s}}_{3,1}+({{\bf h}_{1,{\rm R}}^{1}}^{\!\!*}[5]{\bf { v}}_{4,1}[5]){{s}}_{4,1}}_{L_{1,{\rm SI}}[5]} \nonumber \\
+&\underbrace{({{\bf h}_{1,{\rm R}}^{1}}^{\!\!*}[5]{\bf {v}}_{2,3}[5]){{ {s}}_{2,3}+({{\bf h}^1_{1,{\rm R}}}^{\!\!*}}[5]{\bf {v}}_{2,4}[5]){{s}}_{2,4}}_{L_{1,{\rm OI}}[5]}. \label{eq:rx1}
\end{align}
As shown in (\ref{eq:rx1}), $L_{1,{\rm D}}[5]$ represents the desired sub-equation as it contains desired information symbols $s_{1,3}$ and $s_{1,4}$. 
The sub-equation $L_{1,{\rm SI}}[5]$ denotes the back propagating self-interference signal from the relay because $s_{3,1}$ and $s_{4,1}$ were previously transmitted by user 1. Last, the sub-equation $L_{1,{\rm OI}}$ implies overheard interference signals intended for user 2. By the aforementioned relay transmission, this interference sub-equation should be the same shape that was observed in time slot 4 by user 1, $L_{1,{\rm OI}}[5]=y_1[4]$.

\subsubsection{Decoding}
We explain the decoding procedure for user 1. First, user 1 eliminates the back propagating self-interference signals $L_{1,{\rm SI}}[5]$ from $y_1[5]$ by using knowledge of the effective channel ${{\bf h}_{1,{\rm R}}^{1}}^{\!\!*}[5]{\bf v}_{3,1}[5]$ and ${{\bf h}_{1,{\rm R}}^{1}}^{\!\!*}[5]{\bf v}_{4,1}[5]$ and the transmitted data symbols $s_{3,1}$ and $s_{4,1}$. Then, user 1 removes the effect of interference $L_{1,{\rm OI}}[5]$ by using the fact that $L_{1,{\rm OI}}[5]=y_1[4]$. 
After canceling the known interference, the concatenated input-output relationship seen by user 1 is
\begin{align}
&\left[\!\!%
\begin{array}{c}
  {y}_1[3] \\
  {y}_1[5]-L_{1,{\rm SI}}[5]-y_{1}[4] \\
\end{array}\!\!%
\right]\\\nonumber 
&=\!\! {\left[%
\begin{array}{cc}
  h_{1,3}[3] & h_{1,4}[3] \\
{{\bf h}_{1,{\rm R}}^{1}}^{\!\!*}[5]{\bf v}_{1,3}[5] & {{\bf h}_{1,{\rm R}}^{1}}^{\!\!*}[5]{\bf v}_{1,4}[5]\\
\end{array}%
\right]}\!\!\!\left[%
\begin{array}{c}
  s_{1,3} \\
  s_{1,4} \\
\end{array}%
\right].\label{eq:input-ouput}
\end{align}
Since precoding vectors, ${\bf v}_{1,3}[5]$ and ${\bf v}_{1,4}[5]$, were constructed independently of the direct channel ${ h}_{1,3}[3]$ and $h_{1,4}[3]$, then, the rank of the effective matrix in (\ref{eq:input-ouput}) is two with probability one. As a result, user 1 decodes two desired symbols $s_{1,3}$ and $s_{1,4}$ with five channel uses. Similarly, the other users decode two desired information symbols by using the same method. Consequently, a total eight data symbols have been delivered in five channel uses in the network, implying that a total $d_{\Sigma}^{\rm TWXC}=\frac{8}{5}$ is achieved.
\end{proof}
%\subsection{Comparision with Other Transmission Methods}
%
%This example clearly shows that efficient utilization of side-information substantially improves the sum-DoF of the fully-connected multi-way network. 

\textit{Remark 3 (Sum-DoF Gains)}: Let us compare our result with the other transmission methods. In the two-pair two-way X channel with a two-antenna relay, one can consider one more approach beyond the TDMA and multi-user MIMO transmission methods. The approach is relay-aided interference alignment \cite{Tian}. Without CSIT at users, it is possible to achieve the sum-DoF of $\frac{4}{3}$ with the idea of relay-aided interference alignment in \cite{Tian} because it allows an exchange of a total of eight symbols within 6 time slots. Since the ST-PNC attains the sum-DoF of $\frac{8}{5}$, we can obtain $60\%$ and $20\%$ better sum-DoF gains over TDMA and relay-aided interference alignment methods by the proposed ST-PNC in this network.

\textit{Remark 4 (A Block Fading Scenario)}: The proposed ST-PNC is extendable in a block fading scenario by using the channel use technique explained in  \cite{Namyoon}. By selecting a set of time slots that  belong to mutually different channel blocks, it is possible to design the space-time precoding matrices for the ST-PNC.

\textit{Remark 5 (CSIT at the Relay)}:  In the fifth time block, the required CSI at the relay is 1) the set of outdated CSI between the users, i.e., $\{h_{41}[1], h_{31}[2], h_{42}[1], h_{32}[2], h_{23}[3], h_{13}[4], h_{24}[3], h_{14}[4]\}$ and 2) the current downlink CSIT between the relay and the users, i.e., $\{{\bf h}_{1,{\rm R}}[5],{\bf h}_{2,{\rm R}}[5],{\bf h}_{3,{\rm R}}[5],{\bf h}_{4,{\rm R}}[5]\}$. These sets of required CSIT are able to be obtained by the feedback links from the users to the relay. One possible method is that the user 1 estimates the channel values $\{h_{13}[4], h_{14}[4]\}$ and ${ {\bf h}_{1,{\rm R}}[5]}$ through the control channels in the fourth and the fifth time blocks, respectively. Using the feedback link, before the data transmission of the fifth time block occurs, the user 1 simultaneously sends $\{h_{13}[4], h_{14}[4]\}$ and ${ {\bf h}_{1,{\rm R}}[5]}$ back to the relay through a dedicated feedback channel, which is a conventional CSI feedback procedure of the LTE-A system for multi-user transmissions. Similarly, the other users perform the same procedures for the channel estimation and feedback. As a result, it is possible to know the set of outdated CSI between the users in addition to the downlink CSIT in the fifth time slot to apply the proposed space-time precoding.

\textit{Remark 6 (Connection with Index Coding \cite{Brik})}: The proposed
transmission methods can be explained through the lens of the index coding algorithms in \cite{Brik, Jafar, Maddah-Ali:12, Namyoon,Yang, Gou, Tandon_alt:13}. Specifically, until the relay has global knowledge of messages in the network,  $M$ users propagate information into the network at each time slot. Since the relay has $M$ antennas, it obtains $M$ information symbols per one time slot and the remaining $K\!-\!M$ other users in receiving mode acquire one equation that has both desired and interfering symbols. When the relay obtains all the messages, it starts to control information flows by sending a useful signal to all users so that each user decodes the desired information symbols efficiently based on their previous knowledge: their transmitted symbols and the received equations.

\section{ Sum-DoF Analysis of Fully-Connected Multi-Way Relay Networks }
 So far, we have explained the key idea of our strategy for multi-way communication in two particular network settings. In this section, to provide a more complete performance characterization, we analyze the sum-DoF for a general class of fully-connected multi-way relay networks in terms of system parameters, chiefly, the number of users $K$ and the number of antennas at relays $\{M_{1},M_2,\ldots,M_L\}$.

%The following theorem establishes an inner bound of the network scaling law in the fully-connected multi-way relay network. 

We devote this section to proving the following theorem.
 \begin{theorem} \label{Theorem3}
Consider the fully-connected multi-way relay channel in which $K\geq 3$ users have a single antenna and $L$ relays have $M_{\ell}$ antennas each. An inner bound on the sum-DoF for this network is 
\begin{align}
&d_{\Sigma}(K, \{M_{\ell}\})\nonumber\\&=\min\left\{ \frac{K}{2},\max\left( \frac{K_1^{\star}}{2},\frac{K_2^{\star}(K_2^{\star}-2)}{2K_2^{\star}-3},\frac{(K_3^{\star})^2}{2K_3^{\star}-1}\right)\right\},
\end{align}
where $K_1^{\star} $, $K_2^{\star} $, and $K_3^{\star} $ are integer values defined such that
\begin{align}
&K_1^{\star} =\left\lfloor \sqrt{\left(\sum_{\ell=1}^LM_{\ell}^2-\frac{3}{4}\right)}+\frac{3}{2}\right \rfloor, ~~K_2^{\star} =\left\lfloor \sqrt{ \sum_{\ell=1}^LM_{\ell}^2  }+2\right\rfloor, \nonumber \\
&K_3^{\star}  =\left\lfloor \sqrt{ \sum_{\ell=1}^LM_{\ell}^2  }+1\right\rfloor.
\end{align}

 \end{theorem}

 \proof See Appendix A. \endproof 

As shown in Theorem \ref{Theorem3}, the achievable sum-DoF is characterized by four different integer values: $K$, $K_1^{\star}$, $K_{2}^{\star}$, and $K_{3}^{\star}$. One notable point is that the sum-DoF is upper bounded by the half of the number of users $\frac{K}{2}$ regardless of the relays' antenna configurations and CSIT at users, which will be explained in the following Corollary \ref{corollary1}. Further, according to the relative difference between the number of users $K$ and the relays' antenna configurations $\sum_{\ell=1}^LM_{\ell}^2$, three different sum-DoF values are obtained by 1) the ST-PNC using interference neutralization, 2) the ST-PNC using interference neutralization and alignment jointly, and 3) the relay-aided interference alignment in \cite{Tian}. Therefore, the maximum value of them provides the inner bound of the sum-DoF for the network. %One interesting observation is that when the squared sum of the relays' antennas $\sum_{\ell=1}^L M_{\ell}^2$ is relatively small than the number of users, the one-way communication protocol based on the relay-aided interference alignment attains a better sum-DoF than the multi-way protocol method using the ST-PNC in the fully-connected multi-way relay network, even if this method never achieves the optimal sum-DoF regardless of the number of relay antennas. On the other hand, the squared sum of the relays' antennas $\sum_{\ell=1}^L M_{\ell}^2$ is large enough 

By leveraging the cut-set bound argument in \cite{Lee_Lim_Chun:10} and Theorem \ref{Theorem3}, we establish a sufficient condition on the relays' antenna configuration to achieve the optimal sum-DoF of the fully-connected multi-way relay network.  
 \begin{corollary} \label{corollary1}
For a given $K$, the optimal sum-DoF is $d_{\Sigma}(K,\{M_{\ell}\})=\frac{K}{2}$ and it is attainable, provided that $\sum_{\ell=1}^L M_{\ell}^2 \geq  (K\!-\!1)(K\!-\!2)\!+\!1$, i.e.,  $K_1^{\star}\geq K$.
\end{corollary}

The proof of Corollary \ref{corollary1} relies on the following lemma that provides a sum-DoF outer bound for the two-way relay channel, which has an equivalence with the multi-way relay network when some users and relays cooperate. We reproduce it next for the sake of completeness.   
\begin{lemma} \label{Lemma1} The sum-DoF of the equivalent two-way relay channel is upper bounded as
\begin{align}
\sum_{i=1,i\neq k}^{K}d_{k,i} + \sum_{k=1,k\neq i}^{K}d_{k,i} \leq 1, \quad \textrm{for} \quad k,i \in\mathcal{U},
\end{align}
where $d_{k,i}=\lim_{P\rightarrow \infty}\frac{R_{k,i}}{\log(P)}.$
\end{lemma}
\proof See \cite{ Lee_Lim_Chun:10}. \endproof

We are ready to prove Corollary \ref{corollary1}.

%The proof of Corollary \ref{corollary1} relies on the following lemma that provides the sum-DoF outer bound of this equivalent two-way relay channel. We reproduce next for the sake of completeness.   
%\begin{lemma} \label{Lemma1} The sum-DoF of the equivalent two-way relay channel for the multi-way relay channel is upper bounded as
%\begin{align}
%\sum_{\ell=1,\ell\neq k}^{K}d_{k,\ell} + \sum_{k=1,k\neq\ell}^{K}d_{k,\ell} \leq 1, \quad \textrm{for} \quad k,\ell \in\mathcal{U}.
%\end{align}
%\end{lemma}
%\proof See \cite{ Lee_Lim_Chun:10}. \endproof
%
%We are ready to prove Corollary \ref{corollary1}.
\proof
The achievability is direct from Theorem \ref{Theorem3}. We need to prove that the sum-DoF of the fully-connected multi-way relay channel cannot be greater than $\frac{K}{2}$ regardless of relay configurations. The key idea of the proof is to apply cut-set bounds for different cooperation scenarios among users in the network. Because cooperation among users or relays does not degrade the DoF of the channel, we consider a cooperation scenario in which all users except for user $k$ fully cooperate with each other by sharing antennas and messages and all relays share antennas to form a virtual relay with $\sum_{\ell=1}^LM_{\ell}$ antennas. Under this cooperation setup, we can equivalently convert the original network into a fully-connected two-way relay channel where the user group has $K-1$ antennas, user $k$ has a single antenna, and a virtual relay has $M_t=\sum_{\ell=1}^LM_{\ell}$ antennas. From Lemma \ref{Lemma1}, by adding all $K$ inequalities, the sum-DoF of the fully-connected multi-way relay channel is upper bounded as
\begin{align}
 \sum_{ \ell\neq k}^{K}\sum_{k=1}^{K}d_{\ell,k} \leq \frac{K}{2}.
\end{align}

 \endproof

\subsection{Special Cases}

To shed further light on the implications of Theorem \ref{Theorem3}, it is instructive to consider certain extreme cases and examples.

\subsubsection{Distributed Relays with a Single Antenna}In this case, all $L$ relays are equipped with a single antenna. This case possibly represents the scenario where multiple  relays with a single antenna each help the multi-way information exchange of other users in a dense (fully-connected) network. By setting $M_{\ell}=1$ for $\ell\in\{1,2,\ldots,L\}$, the sum-DoF is summarized in the following corollary. 

 \begin{corollary} \label{corollary2}
When $M_{\ell}=1$ for $\ell\in\{1,\ldots,L\}$, the sum-DoF is given in (\ref{eq:co2}).
\begin{figure*}
\begin{align}
d_{\Sigma}(K, L)=\min\left\{ \frac{K}{2},\max\left( \frac{\left\lfloor \sqrt{\left(L-\frac{3}{4}\right)}+\frac{3}{2}\right \rfloor}{2},\frac{\left\lfloor \sqrt{L  }+2\right\rfloor \left\lfloor \sqrt{L  }\right\rfloor}{2\left\lfloor \sqrt{L  }+2\right\rfloor-3},\frac{ \left(\left\lfloor \sqrt{ L  }+1\right\rfloor\right)^2}{2\left\lfloor \sqrt{ L  }+1\right\rfloor-1}\right)\right\}.\label{eq:co2}
\end{align}
\end{figure*}
\end{corollary}
This result reveals that for a fixed $K$, the sum-DoF linearly increases with respect to the square root of the number of relays $d_{\Sigma}(K, L) \sim c\sqrt{L}$ for $c>0$ in the regime of $L< (K-1)(K-2)$, the sum-DoF grows slowly. Meanwhile, in the regime of $L\geq (K-1)(K-2)+1$, it is possible to obtain the optimal $\frac{K}{2}$ sum-DoF of the network.

\subsubsection{A Single Relay with Multiple Antennas}
As another extreme case, let us consider the case of a single relay with $M_1$ antennas. This case can correspond to the scenario where $K$ users exchange multi-way messages with the help of a single relay (or a base station) with $M_1$ antennas. In this case, the following corollary provides a simplified sum-DoF expression.
 \begin{corollary} \label{corollary3}
When a single relay has $M_{1}$ antennas, the sum-DoF is given in (\ref{eq:co3}).
\begin{figure*}
\begin{align}
d_{\Sigma}(K, M_1)=\min\left\{ \frac{K}{2},\max\left( \frac{\left\lfloor \sqrt{\left(M_1^2-\frac{3}{4}\right)}+\frac{3}{2}\right \rfloor}{2},\frac{(M_1+2)M_1 }{2 M_1-1},\frac{ \left(M_1+1\right)^2}{2M_1+1}\right)\right\}. \label{eq:co3}
\end{align}
\end{figure*}
\end{corollary}
This shows that the sum-DoF linearly increases with respect to $M_1$ until the optimal sum-DoF is achieved, which different than $M_{\ell}=1$. This benefit comes from the joint processing for the interference management at a relay with multiple antennas, as opposed to the distributed processing at the relays with a single antenna. Further, we recover the sufficient condition for obtaining the optimal sum-DoF derived in our previous work \cite{namyoon:13} given by $M_1 \geq  K-1$.
 \subsection{Sum-DoF Comparison }
In this section, we compare the achievable sum-DoF in Theorem \ref{Theorem3} with that obtained by a generalization of the orthogonalize-and-forward method in \cite{Rankov, Lee_Lee_Lee:13}. %and relay-aided interference alignment in \cite{Tian}.
%
%\subsubsection{A Generalized Orthogonalize-and-Forward (G-OF) Transmission}
The generalized-OF (G-OF) relaying strategy does not exploit direct paths between users, which is similar to the conventional OF methods in \cite{Rankov, Lee_Lee_Lee:13}. The difference is that, instead of using the pair-wise information exchange protocol as in \cite{Rankov,Lee_Lee_Lee:13}, we apply the multi-way information exchange protocol so that all users can exchange information symbols in a cyclic manner in the network. The following lemma yields the achievable sum-DoF by the G-OF method.
\begin{lemma} \label{Lemma2} The achievable sum-DoF by the G-OF is
\begin{align}
d^{{\rm G-OF}}_{\Sigma}(K, \{M_{\ell}\})=\frac{ \min\left\{ K, \left\lfloor \sqrt{\left(\sum_{\ell=1}^LM_{\ell}^2 \right)}+1\right \rfloor \right\}}{2}.
\end{align}
\end{lemma}
\proof
Proof is direct from \cite{Rankov,Lee_Lee_Lee:13} and Theorem 3.
\endproof

\begin{figure}
\centering
\includegraphics[width=3.7in]{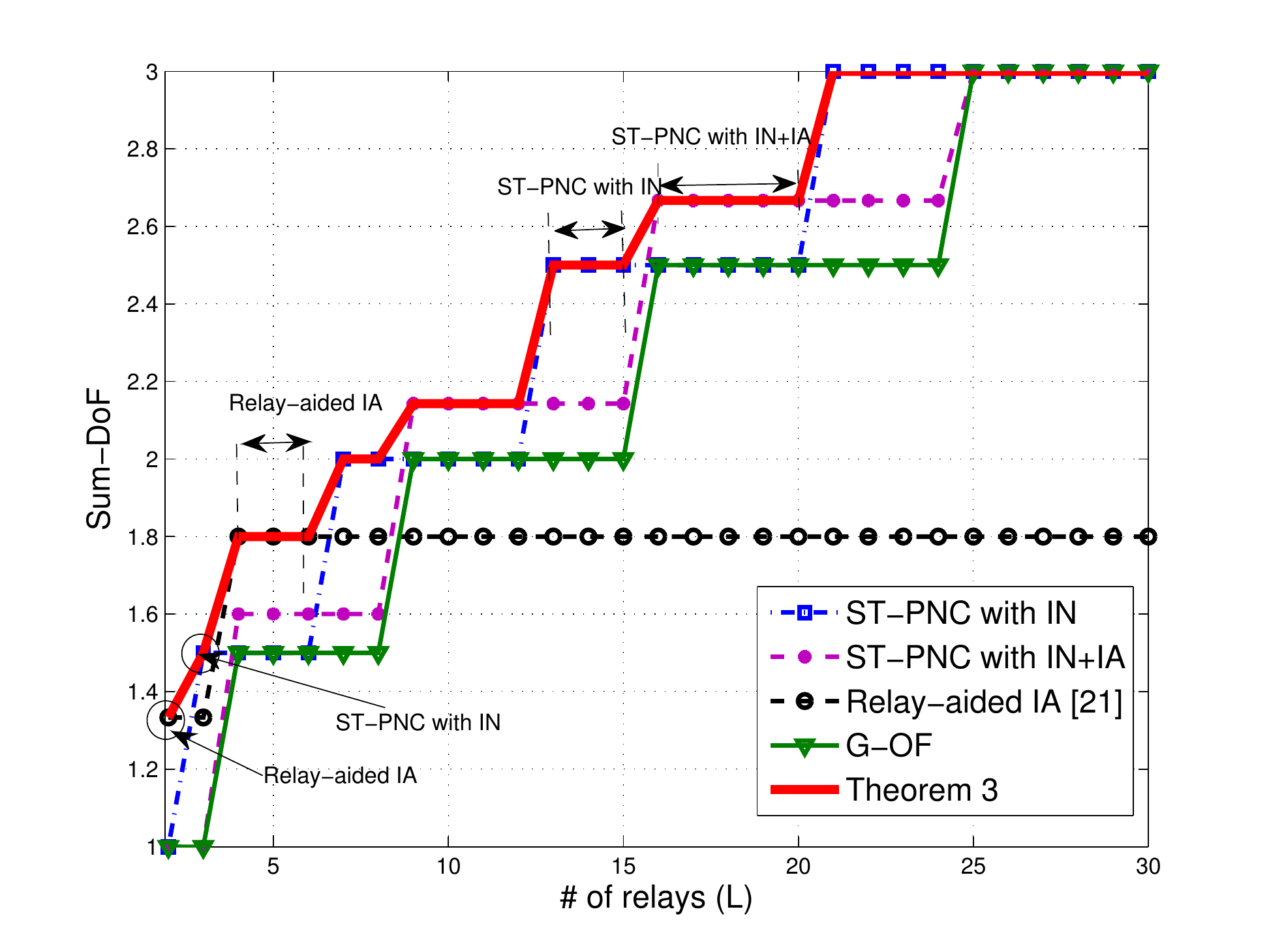}\vspace{-0.1cm}
\caption{Sum-DoF comparision bewteen the proposed ST-PNC and the G-OF method when $K=6$ and $L$ relays have a single antenna.} \label{fig:DoF}\vspace{-0.1cm}
\end{figure}

Fig. \ref{fig:DoF} illustrates the achievable sum-DoF regions achieved by the proposed ST-PNC, relay-aided interference alignment, and the G-OF method as a function of the number of single-antenna relays $L$ when $K=6$. As $L$ increases, the sum-DoF improves with the scale of $\sqrt{L}$ approximately, which agrees with Corollary \ref{corollary2}. One interesting observation is that the proposed ST-PNC always provides a better sum-DoF than the G-OF method in the regime of $L\leq 25$. This DoF gain comes from exploiting additional side-information given by the direct links in addition to the self-interference signals. For the specific values of $L\in\{2,4,5,6\}$, the relay-aided interference alignment \cite{Tian} provides a better sum-DoF than other methods, even if it never achieves the cut-set bound regardless of ${L}$. This reveals that when the number of relays is limited such that they cannot manage multi-way information flows, it is better to communicate through direct links with a one-way protocol instead of using a multi-way protocol. Whereas, when the number of relays is large enough to control the multi-directional information flows, i.e., $L\geq 7$, the multi-way communication protocol in conjunction with the proposed ST-PNC attains a better sum-DoF of the network.

\subsection{An Achievable Rate Computation} 
So far, we have ignored the effects of noise to make the explanations clear. In this subsection, we analyze the achievable rate of the proposed ST-PNC considering noise to show how the ST-PNC behaves in the finite SNR regime by focusing on the two-way interference channel with a MIMO relay in Section III. 

Consider an information flow from user 3 to user 1. Under the premise that the relay applies a decode-and-forward (DF) strategy, an achievable rate is derived for the information transfer of the symbol $s_{1,3}$. To do this, we first compute an achievable rate from the user 3 to the relay. In the second time slot, the relay is able to decode the information symbol $s_{1,3}$ using a ZF decoder with the rate 
\begin{align}
R_{1,3}[2]=\log_2\left(1+\frac{P}{\sigma^2}|{\bf u}_{1,3}^*[2]{\bf h}^1_{{\rm R},3}[2]|^2\right)
\end{align}
where ${\bf u}_{1,3}^*[2]\in {\rm null}\left({\bf h}^1_{{\rm R},4}[2]\right)$ and $\|{\bf u}_{1,3}^*[2]\|_2=1$. 

Next, we compute the information transfer rate from the relay to the user 1. With an uniform power allocation strategy, in the third time slot, the relay sends a linear combination of the four decoded symbols to the users. Normalizing the received signals in the second and the third time slot by multiplying $\frac{1}{\sqrt{P}}$ and $\frac{2}{\sqrt{P}}$, from (\ref{eq:rec_two_Ic}), the resulting input-output relationship with noise is given by
\begin{align}
&\underbrace{\left[\!\!\!%
\begin{array}{c}
 \frac{1}{\sqrt{P}} {y}_{1}[2]\\
 \frac{2}{\sqrt{P}}{y}_{1}[3]- {{\bf h}_{1,{\rm R}}^{1}}^{\!\!*}[3]{\bf { v}}_{3,1}[3]{ s}_{3,1}\\
\end{array}%
\!\!\!\right]}_{{\bf \tilde y}_1}\!\!\nonumber\\
&=\!\!
\underbrace{\left[%
\begin{array}{cc}
  {h}_{1,3}[2] & { h}_{1,4}[2] \\
 {{\bf h}_{1,{\rm R}}^{1}}^{\!\!*}[3]{\bf {v}}_{1,3}[3] &  {\bf h}^{1}_{1,\textrm{R}}[3]^*{\bf v}_{2,4}[3]\\
\end{array}%
\right]}_{[{\bf \tilde{h}}_{1,3}~{\bf \tilde{h}}_{2,4}]}\left[%
\begin{array}{c}
  { s}_{1,3} \\
 { s}_{2,4} \\
\end{array}%
\!\!\right] \!\!+\!\! \underbrace{ \left[%
\begin{array}{c}
  \frac{1}{\sqrt{P}}{z}_{1}[2]\\
 \frac{2}{\sqrt{P}}{z}_{1}[3]\\
\end{array}%
\!\!\right]}_{{\bf \tilde z}_1}\!\!.\label{eq:rec_two_IC_v2}
\end{align} 
To decode the symbol $s_{1,3}$, we apply a successive interference cancellation method that decodes the interference symbol ${s}_{2,4}$ first and then eliminates the interference effect from the received signal vector ${\bf \tilde y}_1$ in (\ref{eq:rec_two_IC_v2}), namely,
\begin{align}
{\bf \tilde y}_1-{\bf \tilde h}_{2,4}s_{2,4}={\bf \tilde h}_{1,3}s_{1,3}+{\bf \tilde z}_1.
\end{align}
Applying maximum ratio combining technique ${\bf u}_{1,3}^*[3]=\frac{{\bf \tilde h}_{1,3}^*}{\|{\bf \tilde h}_{1,3}\|_2}$, which provides the effective noise power $\mathbb{E}[|{\bf u}_{1,3}^*[3]{\bf \tilde z}_1|^2]=\frac{5\sigma^2}{2P}$, the achievable rate for the information transfer of the symbol $s_{1,3}$ in the third time slot is given by
\begin{align}
R_{1,3}[3]= \log_2\left(1+\frac{P}{2.5\sigma^2}\|{\bf \tilde h}_{1,3}\|_2^2\right).
\end{align}
Since the DF method is used, the ergodic achievable rate for the information symbol $s_{1,3}$ is
\begin{align}
\bar{R}_{1,3} =\mathbb{E}\left[\min\left\{R_{1,3}[2],R_{1,3}[3]\right\}\right]
\end{align}
where the expectation is taken over the all fading channels associated with the rate computation. By symmetry, the sum of the ergodic rate emerges as
\begin{align}
R_{\Sigma}^{{\rm STPNC}}=\bar{R}_{1,3}+\bar{R}_{2,4}+\bar{R}_{3,1}+\bar{R}_{4,2}=\frac{4}{3}\bar{R}_{1,3}.
\end{align}
Using TDMA, due to symmetry of the network, the achievable sum rate is given by
\begin{align}
R_{\Sigma}^{{\rm TDMA}}=\mathbb{E}_{h_{1,3}[2]}\left[\log_2\left(1+ \frac{P}{\sigma^2} |h_{1,3}[2]|^2\right)\right].
\end{align}

\begin{figure}
\centering
\includegraphics[width=3.7in]{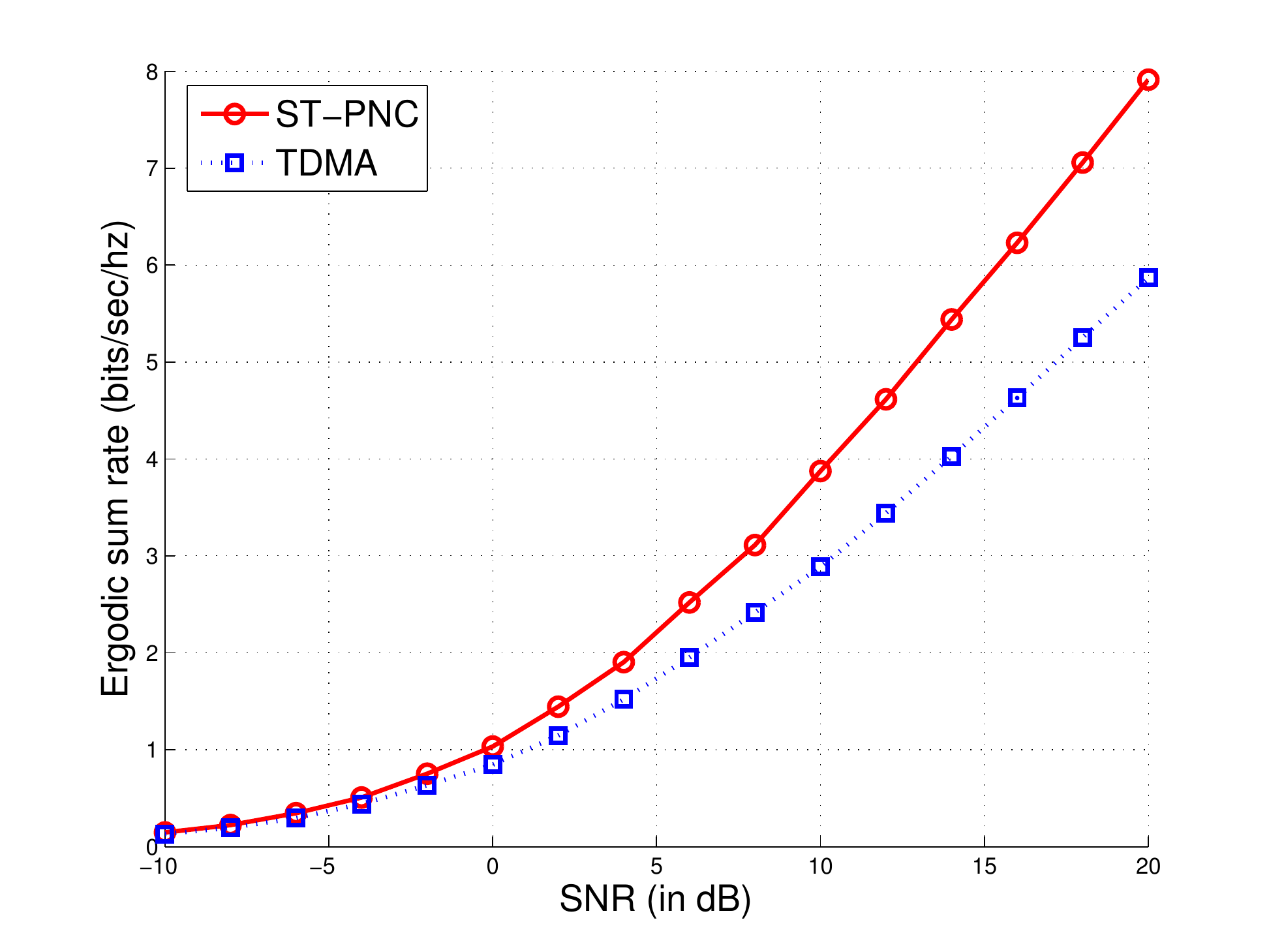}\vspace{-0.1cm}
\caption{The ergodic sum rate comparision bewteen the proposed ST-PNC method and the TDMA method for the two-way interference channel with a MIMO relay.}  \label{fig:sum_rate}\vspace{-0.1cm}
\end{figure}

Fig. \ref{fig:sum_rate} illustrates the ergodic sum-rate obtained by TDMA and proposed method when each channel is drawn from the complex Gaussian distribution, i.e., $\mathcal{CN}(0,1)$.The proposed transmission method provides a better sum rate than TDMA when the SNR is larger than 8 dB. This superior performance in the finite SNR regime is because the ST-PNC makes it possible to obtain signal diversity gain from the direct and detoured links. It is worth noting that the achievable rate derived in this subsection can be further improved by optimizing the power allocation strategy at the relay or by using an advanced decoding strategy in the second time slot, which would be an interesting extension in future work.

\section{Conclusion}
We presented a new physical-layer network coding method called space-time physical-layer network coding (ST-PNC). We used it to establish inner bounds on the sum-DoF of fully-connected multi-way relay network in terms of the number of users, the number or relays, and the number of antennas at each relay. The key idea of  ST-PNC is to control information flows so that each user can exploit overheard interference signals as side-information in addition to what it sent previously. We have demonstrated the superiority of this approach in a sum-DoF sense compared to previously known interference management strategies. Our key finding is that efficiently exploiting interference signals as side-information leads to substantial performance improvements of fully-connected multi-way relay networks.

An interesting direction for future study would be to investigate the effects of having an asymmetric number of antennas at the users, partial network connectivity, full-duplexing operation, and channel knowledge at users in terms of the sum-DoF of the network. Further, with the different relay operations such as amplify-and-forward, decode-and-forward, denoise-and-forward, and compute-and-forward, it would also be interesting to characterize the achievable rate regions the fully-connected multi-way relay networks. %

\appendices
\section{Proof of Theorem 3}
We prove Theorem \ref{Theorem3} using the proposed ST-PNC and relay-aided interference alignment \cite{Tian} according to three different network configurations: 1) $\sum_{\ell=1}^L M_{\ell}^2 \geq (K_1-1)(K_1-2)+1$, 2) $\sum_{\ell=1}^L M_{\ell}^2 \geq (K_2-2)(K_2-2)$, and 3) $\sum_{\ell}^LM_{\ell}^2\geq \left( \left\lfloor \frac{K_3}{2}\right\rfloor-1\right)^2$. Here, $K_1\leq K$, $K_2\leq K$, and $K_3\leq K$ represent the number of users satisfying three network configurations, respectively.

\subsection{Case of $\sum_{\ell=1}^L M_{\ell}^2 \geq (K_1-1)(K_1-2)+1$}

Consider a sub-network in which we select $K_1$ users from a total of $K$ users, i.e. $K_1\leq K$ such that the sum of relays' antennas satisfies the condition of $ \sum_{\ell=1}^LM_{\ell}^2 \geq (K_1-1)(K_1-2)+1$. In this sub-network, we demonstrate that $K_1$ users can exchange $K_1-1$ independent messages with each other by using $|\mathcal{T}_1|=K_1$ time slots for phase one and $|\mathcal{T}_2|=K_1-2$ time slots for phase two, thereby providing the sum-DoF of $\frac{K_1(K_1-1)}{2K_1-2}=\frac{K_1}{2}$. 

%Among $K$ users, $\eta\leq K$ users participate in communication, 
%For a given relay antenna resources $ \sum_{\ell=1}^LM_{\ell}^2$
%
%This network configuration corresponds to the case where the resource of relays' antennas are sufficient to control multi-way information flows. Under the premise of $\left\lfloor \frac{\sum_{\ell=1}^LM_{\ell}^2}{K-2}\right\rfloor \geq K-1$, we will demonstrate that all $K$ users can exchange $K-1$ independent messages with each other within $2K-2$ time slots, thereby attaining  a total of $\frac{K}{2}$ sum-DoF. Later, in fact, we will prove that this sum-DoF is optimal, i.e., matched with an information-theoretic cut-set outer bounds. %ST-PNC involves two phases: side-information learning and space-time relay transmission. 

We start with the side-information learning phase. This phase spans $K_1$ time slots, $\mathcal{T}_1=\{1,2,\ldots, K_1\}$. In time slot $k\in \mathcal{T}_1$, user $i \in \mathcal{S}_k=\{1,2,\ldots,K_1\}/\{k\}$ sends signal $x_{i}[k]=s_{k,i}$ for user $k\in\mathcal{D}_k=\{k\}$.The received signals at user $k$ and the $\ell$-th relay are 
\begin{align}
{ y}_{k}[k] &=\sum_{i \in \mathcal{S}_k}{h}_{k,i}[k]s_{k,i}  \label{eq:dis2} \\
{\bf y}^{\ell}_{{\rm R}}[k]&=\sum_{i \in \mathcal{S}_k}{\bf h}^{\ell}_{{\rm R},i}[k]s_{k,i},  \quad \ell\in\{1,2,\ldots,L\} \label{eq:dis3}
\end{align} 
Through phase one, user $k$ acquires a linear equation that contains $K_1-1$ desired information symbols. Further, since relay $\ell$ has $M_{\ell}$ antennas, it obtains a total of $M_{\ell}(K_1\!-\!1)$ linear equations that contain a total of $K_1(K_1-1)$ information symbols in the network. These overheard equations at $L$ relays will be propagated in the second phase while controlling information flows.
 
For the space-time relay transmissions, we employ $K_1-2$ time slots, $\mathcal{T}_2=\{K_1+1,K_1+2,\ldots,2K_1-2\}$, for the second phase. In this phase, $L$ relays cooperatively send out linear combinations of received signals during the previous phase by applying the proposed space-time relay transmission. Specifically, in time slot $t\in \mathcal{T}_2$, $L$ relays cooperatively send the received signal vectors $\{ {\bf y}_{\rm R}^1[k],{\bf y}_{\rm R}^2[k],\ldots,{\bf y}_{\rm R}^L[k]\}$ in $k\in \mathcal{T}_1$ using precoding matrices $\{{\bf V}_{\rm R}^1[t,k],{\bf V}_{\rm R}^2[t,k],\ldots,{\bf V}_{\rm R}^L[t,k]\}$. Then, the transmitted signal vector of relay $\ell \in \mathcal{R}$ in time slot $t\in \mathcal{T}_2$ is 
\begin{align}
{\bf x}^{\ell}_{\rm R}[t]&=\sum_{k=1}^{K_1}{\bf V}^{\ell}_{\rm R}[t,k]{\bf y}_{\rm R}^{\ell}[k] \nonumber  \\
&=\sum_{k=1}^{K_1}{\bf V}^{\ell}_{\rm R}[t,k]\left(\sum_{i \in \mathcal{S}_k}{\bf h}^{\ell}_{{\rm R},i}[k]s_{k,i}\right).
 \end{align}
Then, the received signal at user $j \in \mathcal{U}$ in time slot $t\in \mathcal{T}_2$ is given by
\begin{align}
{y}_{j}[t]=& \sum_{\ell=1}^L{{\bf h}^{\ell}_{j,{\rm R}}}^{\!\!\!\!*}[t]{\bf x}^{\ell}_{{\rm R}}[t] \\
=& \sum_{\ell=1}^L{{\bf h}^{\ell}_{j,{\rm R}}}^{\!\!\!\!*}[t]\left[\sum_{k=1}^{K_1}{\bf V}^{\ell}_{\rm R}[t,k]\left(\sum_{i \in \mathcal{S}_k}{\bf h}^{\ell}_{{\rm R},i}[k]s_{k,i}\right)\right]\\
=&\sum_{k=1}^{K_1}\sum_{i \in \mathcal{S}_k} \sum_{\ell=1}^L{{\bf h}^{\ell}_{j,{\rm R}}}^{\!\!\!\!*}[t]{\bf V}^{\ell}_{\rm R}[t,k] {\bf h}^{\ell}_{{\rm R},i}[k] s_{k,i},% \\
%=&\sum_{k=1}^K\sum_{i \in \mathcal{S}_k} \sum_{\ell=1}^Lp^{\ell}[t]{{\bf h}^{\ell}_{j,{\rm R}}[t]}^*{\bf V}^{\ell}_{\rm R}[k] {\bf h}^{\ell}_{{\rm R},i}[k] s_{k,i}+ \tilde{z}_j[t]
\label{eq:rev_user_j}
\end{align}
where the last equality follows from changing the summation orders.

%The intuition behind this interference neutralization condition is that the relays forward the superpositions of the received information symbols obtained during time slot $k\in \mathcal{T}_1$ through the network at time slot $t\in \mathcal{T}_2$ so that users $p\in\{1,2,\ldots,K\}/\{k,i\}$ does not receive any interference from the relay transmission excepting for self-interference. 

The crux of the space-time relay transmission is to manage multi-directional information flows so that each user does not receive irresolvable interference signals. User $j\in \{1,2,\ldots,K_1\}$ desires to decode $K_1-1$ information symbols $\{s_{j,1},\ldots, s_{j,j-1},s_{j,j+1},\ldots,s_{j,K_1}\}$ and has knowledge of $K_1-1$ information symbols $\{s_{1,j},\ldots, s_{j-1,j},s_{j+1,j},\ldots,s_{K_1,j}\}$ as side-information. Therefore, $L$ relays cooperatively neutralize $(K_1-1)(K_1-2)$ interference signals over the air so that user $j$ is protected from unmanageable interference signals. To accomplish this, we construct the space-time relay matrices applied at relays across time slot $k\in \mathcal{T}_1$ and $t\in \mathcal{T}_2$ such that 
\begin{align}
 \sum_{\ell=1}^L{{\bf h}^{\ell}_{j,{\rm R}}}^{\!\!\!\!*}[t]{\bf V}^{\ell}_{\rm R}[t,k] {\bf h}^{\ell}_{{\rm R},i}[k]=0, \label{eq:IN_cond_gen}
\end{align}
where $k\in\{1,2,\ldots,K_1\}=\mathcal{T}_1$, $i\in \mathcal{S}_k$, $i\neq j$, and $k\neq j$. Using Tensor product operation property $\textrm{vec}({\bf AXB})=({\bf B}^T\otimes{\bf
A})\textrm{vec}({\bf X})$, we rewrite the interference neutralization condition in (\ref{eq:IN_cond_gen}) in a vector form, which yields
\begin{align}
 \sum_{\ell=1}^L{{\bf g}^{\ell}_{j,{\rm R},i}}^{\!\!\!*}[t,k] {\bf f}_{\rm R}^{\ell}[t,k]=0, \label{IN}
\end{align}
where ${{\bf g}^{\ell}_{j,{\rm R},i}}^{\!\!\!*}[t,k] ={{\bf h}^{\ell }_{{\rm R},i}[k]}^T\otimes {{\bf h}^{\ell}_{j,{\rm R}}}^{\!\!\!*}[t] \in \mathbb{C}^{1\times M_{\ell}^2}$ denotes the effective channel from user $i$ to user $j$ via relay $\ell$ across time slots $t\in \mathcal{T}_2$ and $k \in \mathcal{T}_1$ and ${\bf f}_{\rm R}^{\ell}[t,k]$ denotes the correspoding vector representation of ${\bf V}^{\ell}_{\rm R} [t,k]$, ${\bf f}_{\rm R}^{\ell}[t,k]={\rm vec}\left({\bf V}^{\ell}_{\rm R} [t,k] \right) \in \mathbb{C}^{M_{\ell}^2\times 1}$. For example, in the first time slot, $k=1$, $L$ relays overhear the linear combinations of $\{s_{1,2},\ldots,s_{1,K_1}\}$ and propagate the linear combinations of them in time slot $t \in \mathcal{T}_2$ using precoding vectors $\left\{ {\bf f}_{\rm R}^{1}[t,1],\ldots, {\bf f}_{\rm R}^{L}[t,1]\right\}$ such that information symbol $s_{1,i}$ for $i\in\{2,3,\ldots,K_1\}$ does not reach to user $q\in\{1,2\ldots,K_1\}/\{1,j\}$. To this end we need to jointly design ${\bf f}_{\rm R}^{\ell}[t,1]$ for $\ell\in\{1,\ldots,L\}$ to satisfy the following interference neutralization condition:
\begin{align}&\small
\underbrace{\left[\!\!\!\!\!%
\begin{array}{cccc}
  {{\bf g}^{1}_{3,{\rm R},2} }^{\!\!\!*}[t,1]\!&\! {{\bf g}^{2}_{3,{\rm R},2}}^{\!\!\!*}[t,1]  \!&\! \cdots, \!&\! {{\bf g}^{L}_{3,{\rm R},2}}^{\!\!\!*}[t,1] \\
\vdots & \vdots&\vdots&\vdots\\
  {{\bf g}^{1}_{K_1,{\rm R},2} }^{\!\!\!*}[t,1]\!&\! {{\bf g}^{2}_{K_1,{\rm R},2}  }^{\!\!\!*}[t,1] \!&\! \cdots, \!&\! {{\bf g}^{L}_{K_1,{\rm R},2} }^{\!\!\!*}[t,1]\\\hline
  {{\bf g}^{1}_{2,{\rm R},3}  }^{\!\!\!*}[t,1]\!&\! {{\bf g}^{2}_{2,{\rm R},3} }^{\!\!\!*} [t,1]\!&\! \cdots, \!&\! {{\bf g}^{L}_{2,{\rm R},3} }^{\!\!\!*}[t,1]\\
  {{\bf g}^{1}_{4,{\rm R},3} ] }^{\!\!\!*}[t,1]\!&\! {{\bf g}^{2}_{4,{\rm R},3}  }^{\!\!\!*}[t,1] \!&\! \cdots, \!&\! {{\bf g}^{L}_{4,{\rm R},3}  }^{\!\!\!*}[t,1]\\
\vdots & \vdots&\vdots&\vdots\\
  {{\bf g}^{1}_{K_1,{\rm R},3} }^{\!\!\!*}[t,1]\!&\! {{\bf g}^{2}_{K,{\rm R},3} }^{\!\!\!*}[t,1] \!&\! \cdots, \!&\! {{\bf g}^{L}_{K_1,{\rm R},3}  }^{\!\!\!*}[t,1]\\\hline
\vdots & \vdots&\vdots&\vdots\\ \hline
  {{\bf g}^{1}_{2,{\rm R},K_1}}^{\!\!\!*}[t,1]\!&\! {{\bf g}^{2}_{2,{\rm R},K_1} }^{\!\!\!*}[t,1] \!&\! \cdots, \!&\! {{\bf g}^{L}_{2,{\rm R},K} }^{\!\!\!*}[t,1]\\
  {{\bf g}^{1}_{3,{\rm R},K_1}  }^{\!\!\!*}[t,1]\!&\! {{\bf g}^{2}_{3,{\rm R},K_1}  }^{\!\!\!*}[t,1] \!&\! \cdots, \!&\! {{\bf g}^{L}_{3,{\rm R},K} }^{\!\!\!*}[t,1]\\
\vdots & \vdots&\vdots&\vdots\\
  {{\bf g}^{1}_{K_1\!-\!1,{\rm R},K_1} }^{\!\!\!\!\!*}[t,1]\!&\! {{\bf g}^{2}_{K_1\!-\!1,{\rm R},K_1} }^{\!\!\!\!\!*}[t,1] \!&\! \!\!\cdots, \!&\!\!\!\! {{\bf g}^{L}_{K_1\!-\!1,{\rm R},K_1} }^{\!\!\!\!\!*}[t,1]
\end{array}%
\!\!\!\!\right]}_{(K_1-1)(K_1-2) \times \sum_{\ell=1}^L M_{\ell}^2}\!\! \!\underbrace{\left[\!\!\!\!%
\begin{array}{c}
 {\bf f}_{\rm R}^{1}[t,1] \\
  {\bf f}_{\rm R}^{2}[t,1] \\
  \vdots \\
 {\bf f}_{\rm R}^{L}[t,1] \\
\end{array}%
\!\!\!\!\right]}_{{\bf f}_{\rm R}[t,1]} \nonumber\\&= {\bf 0}.\label{IN_1}
\end{align}
Since all elements of ${{\bf g}^{\ell}_{j,{\rm R},i}}^{\!\!\!*}[t,1]$ are the product of two IID continuous random variables, they are mutually independent. Further, since $\sum_{\ell=1}^LM_{\ell}^2 \geq (K_1-1)(K_1-2)+1$, it is possible to find ${\bf f}_{\rm R}[t,1]$ in the null space of the concatenated channel matrix in (\ref{IN_1}) almost surely. Applying the same principle, for the other time slots $k\in\{2,\ldots, K_1\}$, we construct space-time relay transmission matrices, ${\bf f}_{\rm R}^{\ell}[t,k]={\rm vec}\left({\bf V}^{\ell}_{{\rm R}}[t,k] \right)$, which guarantee the interference neutralization conditions in (\ref{eq:IN_cond_gen}).

Let us explain a decoding method for user $k \in\mathcal{U}$. From the interference neutralization conditions in (\ref{IN}), in every time slot $t\in \mathcal{T}_2$, user $k$ receives one equation that contains $K-1$ desired symbols $\{s_{k,1},\ldots,s_{k,k-1},s_{k,k+1}\ldots,s_{k,K_1}\}$ and $K-1$ self-interference symbols $\{s_{1,k},\ldots,s_{k-1,k},s_{k+1,k}\ldots,s_{K_1,k}\}$,   
\begin{align}
{y}_{k}[t]%&\sum_{k=1}^K\sum_{i \in \mathcal{S}_1} \sum_{\ell=1}^L{{\bf g}^{\ell}_{j,{\rm R},i}[t,k] }^*{\bf f}_{\rm R}^{\ell}[k]s_{k,i}+\tilde{z}_1[t] \\
=&\sum_{i=1,i\neq k}^{K_1}  \sum_{\ell=1}^L{{\bf h}^{\ell}_{k,{\rm R}}}^{\!\!\!*}[t] {\bf V}^{\ell}_{\rm R}[t,k]{\bf h}^{\ell}_{{\rm R},i}[k]s_{k,i} \nonumber \\
&+ \underbrace{\sum_{i=1,i\neq k}^{K_1}   \sum_{\ell=1}^L{{\bf h}^{\ell}_{k,{\rm R}} }^{\!\!\!*}[t]{\bf V}^{\ell}_{\rm R}[t,i]{\bf h}^{\ell}_{{\rm R},k}[i]s_{i,k}}_{L_{k,\rm{SI}}[t]}.
\end{align}
We subtract the contribution of known signals as ${y}_{1}[t]- L_{1,\rm{SI}}[t]$ during the second phase $t\in \mathcal{T}_2$ with $|\mathcal{T}_2|=K_1-2$, which provides $K_1-2$ desired equations for user $k$. Since user $k$ already obtained one equation for desired symbols in phase one, by concatenating all $K_1-1$ received signals obtained over two phases, we obtain the aggregated input-out relationship in a matrix form,
\begin{align}\small
&\underbrace{\left[\!\!\!\!%
\begin{array}{c}
  {y}_k[k] \\
  { y}_k[K_1\!+\!1]\!-\! L_{k,\rm{SI}}[K_1\!+\!1]\\
  \vdots \\
  { y}_k[2K_1\!-\!2]\!-\! L_{k,\rm{SI}}[2K_1\!-\!2] \\
\end{array}\!\!\!\!%
\right]}_{{\bf \tilde{y}}_k}\!\!\!\nonumber \\
&=\small \!\!\underbrace{\left[\!\!\!\!%
\begin{array}{ccccc}
     h_{k,1}[k] \!\!\!&\!\!\! h_{k,2}[k] \!\!\!&\!\!\! \cdots \!\!\!&\!\!\! h_{k,K}[1]\\
     {\tilde h}_{k,1}[K_1\!+\!1] \!\!&\!\! {\tilde h}_{k,2}[K_1\!+\!1] \!\!&\!\! \cdots \!\!&\!\! {\tilde h}_{k,K_1}[K_1\!+\!1]\\
\vdots & \vdots & \ddots & \vdots  \\
   {\tilde h}_{k,1}[2K_1\!-\!2] & {\tilde h}_{k,2}[2K_1\!-\!2] & \cdots & {\tilde h}_{k,K_1}[2K_1\!-\!2]\\
\end{array}\!\!\!\!%
\right]}_{{\bf \tilde{H}}_{k}}\!\!\underbrace{\left[\!\!\!\!%
\begin{array}{c}
  s_{k,1} \\
  s_{k,2} \\
\vdots \\
s_{k,K_1} \\
\end{array}\!\!%
\!\!\right]}_{{\bf s}_k}%!+\!\!\underbrace{\left[\!\!\!%
%\begin{array}{c}
%  z_k[k] \\
%{\tilde z}_k[K\!+\!1]\\
%  \vdots \\
%  {\tilde z}_k[2K\!-\!2] \\
%\end{array}\!\!\!%
%\right]}_{{\bf \tilde{z}}_k} \label{eq:rx1_final_rec}% \nonumber\\
%{\bf \tilde{y}}_1&={\bf \tilde{H}}_{1}{\bf s}_1+{\bf \tilde{z}}_1 ,
\end{align}
where ${\tilde h}_{k,i}[t]=\sum_{\ell=1}^L{{\bf h}^{\ell}_{k,{\rm R}}}^{\!\!\!*}[t]{\bf V}^{\ell}_{\rm R}[t,k]{\bf h}^{\ell}_{{\rm R},i}[k]$ for $t\in \mathcal{T}_2$ denotes the effective channel coefficient from user $k$ to user 1 via $L$ relays in time slot $t\in \mathcal{T}_2$. Since we have used $|\mathcal{T}_1|+|\mathcal{T}_2|=2K_1-2$ time slots and ${\rm rank}\left({\bf \tilde{H}}_{k}\right)=K_1-1$ almost surely, user $k$ obtains $\frac{K_1-1}{2K_1-2}$ sum-DoF. By symmetry, the other users also attain the same sum-DoF. As a result, a total of $\frac{K_1(K_1-1)}{2K_1-2}=\frac{K_1}{2}$ sum-DoF is achievable, provided that $ \sum_{\ell=1}^LM_{\ell}^2 \geq (K_1-1)(K_1-2)+1$ for any $K_1\leq K$. To attain the maximum sum-DoF for the given relays' antenna configuration, we find the maximum positive integer value of $K_1$ satisfying the inequality $ \sum_{\ell=1}^LM_{\ell}^2 \geq (K_1-1)(K_1-2)+1$, which yields,
\begin{align}
K_1^{\star} =\left\lfloor \sqrt{\left(\sum_{\ell=1}^LM_{\ell}^2-\frac{3}{4}\right)}+\frac{3}{2}\right \rfloor.
\end{align}

\subsection{Case of  $\sum_{\ell=1}^L M_{\ell}^2 \geq (K_2-2)^2$}
Suppose a sub-network where $K_2$ users are selected from a total of $K$ users, $K_2\leq K$ such that the sum of the relays' antennas is greater than or equal to $(K_2-2)^2$. In this reduced network, a user intends to send $K_2-2$ independent messages for the other users; thus a total of $K_2(K_2-2)$ independent messages exists. Specifically, user $k$ desires to send the set of information symbols $\{s_{k_1,k},s_{k_2,k},\ldots,s_{k_{K_2},k}\}$ where $k_j$ denotes an index function defined as $k_j=\{(k\!-\!1\!+\!j)\!\!\mod (K_2)\} \!+\!1$. To exchange a total of $K_1(K_1-2)$ information symbols in the reduced network, we spend $|\mathcal{T}_1|=K_2$ and $|\mathcal{T}_2|=K_2-3$ time slots in two phases.

In the phase of side-information learning, we spend $K_2$ time slots, $\mathcal{T}_1=\{1,2, \ldots, K_2\}$. In each time slot of the first phase, $K_2-2$ users transmit information symbols, while the remaining $2$ users receive the linear combination of the transmitted $K_2-2$ symbols. Recall the index function $k_j=\{(k\!-\!1\!+\!j)\!\!\mod (K_2)\} \!+\!1$. With this index function, we define the set of receiving and transmitting users in time slot $k$ as $\mathcal{D}_k=\{k, k_1\}$ and $\mathcal{S}_k=\{k_2,\ldots,k_{K_2-1}\}$, $|\mathcal{D}_k|=2$ and $|\mathcal{S}_k|=K_2-2$. In time slot $k\in \mathcal{T}_1$, two users in $\mathcal{D}_k=\{k,k_1\}$ listen to the signals sent by $K_2-2$ users belonging to the set $\mathcal{S}_k= \{k_2,\ldots,k_{K_2-1}\}$.
When the users in $\mathcal{S}_k$ send information symbols $\{s_{k,k_2},\ldots,s_{k,k_{K_2-1}}\}$ to user $k$ simultaneously, the received signals at user $k$, user $k_1$, and relay $\ell\in\{1,2,\ldots,L\}$ are 
\begin{align}
{y}_{k}[k] &=  \sum_{k_i \in \mathcal{S}_k}{h}_{k,k_{i}}[k]s_{k,k_{i}}, \label{eq:rx_k}\\
{y}_{k_1}[k] &=  \sum_{k_i \in \mathcal{S}_k}{h}_{k_1,k_{i}}[k]s_{k,k_{i}}, \label{eq:rx_k_1}\\
{\bf y}_{\rm{R}}^{\ell}[k] &= \sum_{k_i \in \mathcal{S}_k}{\bf h}_{{\rm R},k_{i}}[k]s_{k,k_{i}}. \label{eq:relay_k}
\end{align}
Note that user $k\in\{1,2,\ldots,K_2\}$ acquires a linear equation consisting of the desired symbols whereas user $k_1$ overhears a linear combination of interfering symbols. Further, the $\ell$-th relay obtains $M_{\ell}K_2$ equations, which contain a total of $K_2(K_2-2)$  information symbols transmitted by the users.

For the second phase, we use $K_2-3$ time slots, $t\in\mathcal{T}_2=\{K_2+1,K_2+2,\ldots,2K_2-3\}$. In this phase, $L$ relays send out linear combinations of the received signals using the space-time relay precoding method. The transmitted signal vector of relay $\ell \in \mathcal{R}$ in time slot $t\in \mathcal{T}_2$ is 
\begin{align}
{\bf x}^{\ell}_{\rm R}[t]&=\sum_{k=1}^{K_2}{\bf V}^{\ell}_{\rm R}[t,k]\sum_{k_i \in \mathcal{S}_k}{\bf h}_{{\rm R},k_{i}}[k]s_{k,k_{i}}
 \end{align}
and the received signal at user $j \in \mathcal{U}$ in time slot $t\in \mathcal{T}_2$ is given by
\begin{align}
{y}_{j}[t]=&\sum_{k=1}^{K_2}\sum_{k_i \in \mathcal{S}_k} \sum_{\ell=1}^L{{\bf h}^{\ell}_{j,{\rm R}}}^{\!\!\!*}[t]{\bf V}^{\ell}_{\rm R}[t,k] {\bf h}^{\ell}_{{\rm R},k_i}[k] s_{k,k_i}.% \\
%=&\sum_{k=1}^K\sum_{i \in \mathcal{S}_k} \sum_{\ell=1}^Lp^{\ell}[t]{{\bf h}^{\ell}_{j,{\rm R}}[t]}^*{\bf V}^{\ell}_{\rm R}[k] {\bf h}^{\ell}_{{\rm R},i}[k] s_{k,i}+ \tilde{z}_j[t]
\label{eq:rev_user_j_d}
\end{align}
Unlike the previous case, in this regime of $\sum_{\ell=1}^{L}M_{\ell}^2 \geq (K_2-2)^2$, the proposed space-time relay precoding method exploits the current CSIT at the $\ell$th relay to the users, ${{\bf h}^{\ell}_{k,\rm{R}}}^{\!\!\!*}[t]$ for $t\in \mathcal{T}_2$ and outdated CSI between users, $\left\{h_{i,j}[k]\right\}$ for $k\in \mathcal{T}_1$ to perform interference alignment and neutralization jointly. 
 
To illustrate, we explain the design principle of $\{ {\bf V}^{1}_{\rm R}[t,k],\ldots,{\bf V}^{L}_{\rm R}[t,k]\}$ carrying $s_{k,k_{i}}$ from an index coding perspective. Recall that  data symbol $s_{k,k_{i}}$ is only desired by user $k$ and it is unmanageable interference to all the other users excepting for user $k_{i}$ (the user who sent $s_{k,k_{i}}$) and user $k_1$ (the user who overheard $s_{k,k_{i}}$ in time slot $k\in \mathcal{T}_1$). This is because user $k_{\ell}$ is able to cancel self-interference using knowledge of $s_{k,k_{i}}$. Further, user $k_1$ can remove the effect of $s_{k,k_{i}}$ from the relay transmission, provided that user $k_1$ receives the same interference shape of $h_{k_1,k_{i}}[k]s_{k,k_{i}}$, which was obtained in time slot $k\in \mathcal{T}_1$ during phase one in the form of $y_{k_1}[k]=\sum_{i=2}^{K_2}h_{k_1,k_{i}}[k]s_{k,k_{i}}$. Meanwhile, information symbol $s_{k,k_{i}}$ is interference to the other users excepting user $k$, user $k_{i}$, and user $k_1$. Using this fact, we design precoding matrices$\{ {\bf V}^{1}_{\rm R}[t,k],\ldots,{\bf V}^{L}_{\rm R}[t,k]\}$ carrying $s_{k,k_{i}}$ so that it does not reach the other users while providing the same interference shape to user $k_1$. This condition is equivalently written as
\begin{align}
& \sum_{\ell=1}^L{{\bf h}^{\ell}_{j,{\rm R}}}^{\!\!\!*}[t]{\bf V}^{\ell}_{\rm R}[t,k] {\bf h}^{\ell}_{{\rm R},k_i}[k]=0  \label{eq:IAIN1} \\
 &\sum_{\ell=1}^L{{\bf h}^{\ell}_{k_1,{\rm R}}}^{\!\!\!*}[t]{\bf V}^{\ell}_{\rm R}[t,k] {\bf h}^{\ell}_{{\rm R},k_i}[k]= h_{k_1,k_{i}}[k],\label{eq:IAIN2}
\end{align}
where $j\in\{1,2,\ldots,K_2\}/\{k,k_1,k_i\}$, $t\in \mathcal{T}_2$, and $k\in \mathcal{T}_1$. With the same argument shown in (\ref{IN}) and (\ref{IN_1}), since $\sum_{\ell=1}^LM_{\ell}^2\geq (K_2-2)^2$, it is possible to construct relay precoding matrices ensuring (\ref{eq:IAIN1}) and (\ref{eq:IAIN2}) almost surely.

From the space-time relay transmission, in the second phase, the received signal at user $k$, $y_k[t]$, for $t\in \mathcal{T}_2$ is represented as the sum of three sub-linear equations: 1) desired equation $L_{k,{\rm D}}[t]$, 2) self-interference equation $L_{k,{\rm SI}}[t]$, and 3) overheard interference equation $L_{k,{\rm OI}}[t]$, 
\begin{align}
{y}_{k}[t]=&  \underbrace{\sum_{k_i \in \mathcal{S}_k}\sum_{\ell=1}^L{{\bf h}^{\ell}_{k,{\rm R}}}^{\!\!\!*}[t]{\bf V}^{\ell}_{\rm R}[t,k] {\bf h}^{\ell}_{{\rm R},k_i}[k] s_{k,k_i} }_{L_{k,{\rm D}}[t]}\nonumber\\
&+\underbrace{\sum_{j=1, j\neq k, j\neq k_1}^{K_2} \sum_{\ell=1}^L{{\bf h}^{\ell}_{k,{\rm R}}}^{\!\!\!*}[t]{\bf V}^{\ell}_{\rm R}[t,j] {\bf h}^{\ell}_{{\rm R},k}[j] s_{j,k}}_{L_{k,{\rm SI}}[t]} \nonumber\\
&+\underbrace{\sum_{i\in \mathcal{S}_{k_{K_2-1}}} \sum_{\ell=1}^L{{\bf h}^{\ell}_{k,{\rm R}}}^{\!\!\!*}[t]{\bf V}^{\ell}_{\rm R}[t,i] {\bf h}^{\ell}_{{\rm R},i}[i] s_{k_{K_2-1},i} }_{L_{k,{\rm OI}}[t]}.% \\
%=&\sum_{k=1}^K\sum_{i \in \mathcal{S}_k} \sum_{\ell=1}^Lp^{\ell}[t]{{\bf h}^{\ell}_{j,{\rm R}}[t]}^*{\bf V}^{\ell}_{\rm R}[k] {\bf h}^{\ell}_{{\rm R},i}[k] s_{k,i}+ \tilde{z}_j[t]
\label{eq:relayafter}
\end{align}
Note that from (\ref{eq:IAIN2}), the overheard interference equation $L_{k,{\rm OI}}[t]$ in the second phase has the same shape as the previously received equation at user $k$ in time slot $k_{K_2-1}$ of phase one, $y_k[k_{K_2-1}]=L_{k,{\rm OI}}[t]$.

Let us explain the decoding procedure for user $k$ for $k\in\{1,2,3,\ldots,K_2\}$. The decoding procedure involves three steps: 1) the cancellation of the back propagating self-interference, $L_{k,{\rm SI}}[t]$, 2) the cancellation of the previously overheard interference, $L_{k,{\rm OI}}[t]$, and 3) the ZF decoding for the desired symbols' extraction. Specifically, user $k$ first removes the effect of the back propagating self-interference $L_{k,{\rm SI}}[t]$ from the observation of $y_k[t]$. Further, $L_{k,{\rm OI}}[t]$ is removed from$y_k[t]$ using the fact that $y_k[k_{K_2-1}]=L_{k,{\rm OI}}[t]$. After canceling the known interference signals, the concatenated input-output relationship seen by user $k$ becomes
\begin{align}
&\left[\!\!\!%
\begin{array}{c}
  {y}_k[k] \\
  {y}_k[K_2\!+\!1]\!-\!{y}_k[ k_{K_2\!-\!1}]\!-\!L_{k,{\rm SI}}[K_2\!+\!1] \\
 \vdots \\
  {y}_k[2K_2\!-\!3]\!-\!{y}_k[ k_{K_2\!-\!1}]\!-\!L_{k,{\rm SI}}[2K_2\!-\!3] \\
\end{array}%
\!\!\!\right]\!\!\nonumber\\
&=\!\!\underbrace{\left[%
\!\!\begin{array}{ccc}
  h_{k,k_2}[k] & \cdots & h_{k,k_{K_2-1}}[k] \\
 \tilde{h}_{k,k_2}[K_2\!+\!1]&  \cdots &\tilde{h}_{k,k_{K_2-1}}[K_2\!+\!1]\\
\vdots & \ddots & \vdots \\
 \tilde{h}_{k,k_2}[2K_2\!-\!3]&  \cdots &\tilde{h}_{k,k_{K_2-1}}[2K_2\!-\!3]\\
\end{array}\!\!\!\!%
\right]}_{{\bf \hat{H}}_{k}}\!\!\!\left[\!\!\!%
\begin{array}{c}
  s_{k,k_2} \\
  s_{k,k_3} \\
  \vdots \\
  s_{k,k_{K_2\!-\!1}} \\
\end{array}%
\!\!\!\right]\!\!, \nonumber
\end{align}
where $\tilde{h}_{k,k_{i}}[t]=\sum_{\ell=1}^L{{\bf h}^{\ell}_{k,{\rm R}}}^{\!\!\!*}[t]{\bf V}^{\ell}_{\rm R}[t,k] {\bf h}^{\ell}_{{\rm R},k_i}[k]$ denotes an effective channel carrying information symbol $s_{k,k_{i}}$ via the relays. Since beamforming matrices, ${\bf V}^{\ell}_{\rm R}[t,k]$ for $t\in \mathcal{T}_2$ were designed independently from the direct channel ${ h}_{k,k_{i}}[k]$ for $k\in \mathcal{T}_1$, then, it follows that $\textrm{rank}\left({\bf \hat{H}}_{k}\right)=K_2-2$. As a result, user $k$ decodes $K_2-2$ desired symbols by using a total of $2K_2-3$ time slots. By symmetry, the other users decode $K_2-2$ desired information symbols by applying the same decoding method. Consequently, the sum-DoF of  $\frac{K_2(K_2-2)}{2K_2-3}$ is achieved. Since this sum-DoF result is true for all $K_2$ such that $\sum_{\ell=1}^LM_{\ell}^2 \geq (K_2-2)^2$, the maximum value of the sum-DoF is obtained when $K_2^{\star} =\left\lfloor \sqrt{ \sum_{\ell=1}^LM_{\ell}^2  }+2\right\rfloor$.

\subsection{Case of  $\sum_{\ell=1}^L M_{\ell}^2 \geq (\left\lfloor \frac{K_3}{2}\right\rfloor-1)^2$}
In this case, we show that the sum-DoF of $\frac{\left(\left\lfloor \frac{K_3}{2}\right\rfloor\right)^2}{K_3-1}$ is achievable by the relay-aided interference alignment \cite{Tian}, which supports one-directional information exchange in the network. Let us consider a $K_3$ user fully-connected multi-way relay network with $L$ relays; each of them has $M_{\ell}$ antennas.  In this network, we consider a partition that separates $\left\lfloor \frac{K_3}{2}\right\rfloor$ users as source nodes and $\left\lfloor \frac{K_3}{2}\right\rfloor$ users  destination nodes, which creates a $\left\lfloor \frac{K_3}{2}\right\rfloor \times \left\lfloor \frac{K_3}{2}\right\rfloor$ $X$ network with the $L$ relays. Then, from the result in \cite{Tian}, we can obtain the sum-DoF of $\frac{\left(\left\lfloor \frac{K_3}{2}\right\rfloor\right)^2}{K_3-1}$ if $\sum_{\ell=1}^L M_{\ell}^2 \geq (\left\lfloor \frac{K_3}{2}\right\rfloor-1)^2$ by using the relay-aided interference alignment. By solving the inequality with respect to $K_3$, we obtain the maximum integer value of $K_3$ as $K_3^{\star}  =2\left\lfloor \sqrt{ \sum_{\ell=1}^LM_{\ell}^2  }+1\right\rfloor$. As a result, the sum-DoF of $\frac{\left(K_3^{\star}\right)^2}{K_3^{\star}-1}$ is achievable, which completes the proof.

\end{document}